\documentclass[11pt]{article}
\usepackage[utf8]{inputenc}
\usepackage[T1]{fontenc}
\usepackage{microtype}
\usepackage[numbib,nottoc]{tocbibind} 
\usepackage[english]{babel}
\usepackage[table,xcdraw]{xcolor}
\usepackage{charter}

\usepackage{csquotes,xpatch} 

\usepackage[svgnames]{xcolor}

\usepackage[shortlabels]{enumitem}

\usepackage{tikz}
\usepackage{pgfplots}
\usetikzlibrary{calc}
\usetikzlibrary{arrows.meta}
\usetikzlibrary{patterns}
\pgfplotsset{compat = 1.18}

\usepackage{amsfonts}
\usepackage{amsmath}
\usepackage{amssymb}
\usepackage{amsthm}
\allowdisplaybreaks[2]

\usepackage{graphicx}
\usepackage{wrapfig}
\usepackage{subcaption}
\usepackage{caption}

\usepackage{mathtools}

\usepackage{bbm}

\usepackage{mystyle}

\usepackage{comment}

\usepackage{hyperref}
\hypersetup{colorlinks=true,
            breaklinks=true,
            linkcolor=[rgb]{0, 0.2, 0.5},
            citecolor=[rgb]{0, 0.35, 0},      
            urlcolor=[rgb]{0, 0.2, 0.5},
            pdfpagemode=UseOutlines}
\usepackage[nameinlink, capitalise]{cleveref}
\usepackage{authblk}  
\usepackage{ifthen}  
\newif\ifExtAbs
\ExtAbsfalse
\ifExtAbs
\else
\fi

\ifExtAbs
  \usepackage[margin=0.5in]{geometry}
\else
  \usepackage[margin=1in]{geometry}
\fi

\theoremstyle{plain}
\newtheorem{theorem}{Theorem}
\newtheorem*{theorem*}{Theorem}
\newtheorem{proposition}[theorem]{Proposition}
\newtheorem*{proposition*}{Proposition}
\newtheorem{corollary}[theorem]{Corollary}
\newtheorem*{corollary*}{Corollary}
\newtheorem{lemma}[theorem]{Lemma}
\newtheorem*{lemma*}{Lemma}

\newtheorem*{conjecture*}{Conjecture}
\newtheorem{result}{Result}
\newtheorem*{result*}{Result}

\newtheorem*{question*}{Question}
\theoremstyle{definition}
\newtheorem{definition}{Definition}
\newtheorem*{definition*}{Definition}

\newtheorem*{example*}{Example}
\theoremstyle{plain}
\newtheorem{remark}{Remark}
\newtheorem*{remark*}{Remark}

\title{
\ifExtAbs
\else
  \vspace{-1cm}
\fi
Optimal Untelegraphable Encryption and \\ Implications for Uncloneable Encryption}
\ifExtAbs
  \author{Anne Broadbent}
\author{Eric Culf}
\author{Denis Rochette}
\affil{}
\else
\author[1]{Anne Broadbent \thanks{anne.broadbent@uottawa.ca}}
\author[2]{Eric Culf \thanks{eculf@uwaterloo.ca}}
\author[1]{Denis Rochette \thanks{denis.rochette@uottawa.ca}}
\affil[1]{Department of Mathematics and Statistics, University of Ottawa}
\affil[2]{Institute for Quantum Computing and University of Waterloo \vspace{-1cm}}

\fi
\date{}

\begin{document}

\maketitle

\ifExtAbs
\ifExtAbs
\else
  \section{Introduction} \label{sec:introduction}
\fi

\emph{Uncloneable encryption} (UE), originally introduced by Broadbent and Lord \cite{BL20}, and named after the no-cloning principle \cite{WZ82,Die82}, is an 
encryption scheme that provides a level of security unachievable classically: no pirating adversary $\advA$, receiving a ciphertext encoding a message $m$ as a quantum state, can apply a quantum operation that produces a bipartite state over registers $\regB$ and~$\regC$ such that two local receiving adversaries, $\advB$ and $\advC$, each accessing one of these registers, can both recover information about $m$, even when given the classical secret~key.
When $m$ is sampled uniformly from the message space and both $\advB$ and $\advC$ are required to recover the entire message, this is referred to as \emph{one-way, or search} security. A more robust notion, called \emph{indistinguishability, or decision} security, allows $\advA$ to fix two messages, receive an encryption of one of them at random, and then have $\advB$ and $\advC$ attempt to guess which one. 
Search security has been achieved in an information-theoretic sense in \cite{BL20}, but indistinguishability security remains an open problem. Existing positive results rely on additional assumptions such as oracles \cite{BL20, AKL+22}, relaxed definitions \cite{KN23, AKY25, CG24,BC25arXiv}, or new conjectures \cite{AB24,BBCBP24}.

Since its introduction, uncloneable encryption has become a fundamental building block of quantum cryptography, supporting applications such as
secure software leasing \cite{KNY20arxiv,BJL+21,ALP21},
uncloneable zero-knowledge proofs \cite{JK24},
quantum copy-protection \cite{AK21,ALL+21,CLLZ21,CMP24},
private-key quantum money \cite{BL20},
uncloneable decryption \cite{GZ20eprint,CLLZ21,SW22arxiv,KT25},
quantum functional encryption \cite{MM24arxiv}.
It has moreover inspired the study of numerous other uncloneable cryptographic primitives.

\emph{Untelegraphable encryption} (UTE), introduced in \cite{CKNY24arxiv}, is a natural relaxation of UE, inspired by the \emph{no-telegraphing principle}, 
 which asserts that, without pre-shared entanglement, quantum information cannot be transmitted through classical channel alone \cite{Wer98}\footnote{The no-telegraphing principle is also referred to as the no-classical-teleportation theorem.}.
The security guarantee for UTE is as follows: any telegraphing adversary $\advA$ who receives a quantum ciphertext of a message cannot generate \emph{classical} information that would enable another adversary, the receiver $\advB$, to recover information about the message even when given the classical secret key. The notions of search and indistinguishability security in UTE are analogous to those defined for~UE. 
We note that early work \cite{BL20} has already established indistinguishable-UE in the oracle model for unentangled adversaries, thus implying  indistinguishable-UTE in the oracle model.  \cite{CKNY24arxiv} improves this to an unconditional construction. They further introduce the notion of collusion-resistant UTE, in which security is preserved even when the adversary~$\advA$ adaptively selects multiple pairs of messages across successive rounds. Constructions of such schemes are obtained from pseudorandom functions, or from pseudorandom states when the number of rounds is bounded. Finally, they present an everlasting collusion-resistant UTE scheme in the quantum random oracle model.

A central reason for studying UTE lies in its conceptual connection to the two foundational no-go theorems already mentionned: the no-cloning theorem and the no-telegraphing theorem.
These two no-go theorems are in fact informationally equivalent~\cite{NZ24}: to construct a copy using a telegrapher, one may simply copy the intermediate classical output produced by the telegrapher’s channel and subsequently apply the reconstruction procedure to each of these classical copies. Conversely, to achieve telegraphing from a cloner, if one can generate a sufficiently large number of copies to enable a full classical characterization of the quantum state (\emph{e.g.}, via quantum state tomography), then this classical description can be transmitted through the telegrapher’s channel to reconstruct the original state. Naturally, this latter direction incurs a significantly higher computational cost. The computational separation between the no-cloning and no-telegraphing principles has been explored in detail in~\cite{NZ24}. The informational equivalence between the no-cloning and no-telegraphing theorems motivates our investigation of UE via UTE.

\ifExtAbs
   \paragraph{Results}
\else
  \subsection{Results}
\fi

Our work investigates the existence of a UTE scheme that achieves private-key untelegraphable-indistinguishability security unconditionally, as well as its implications for UE. 
In particular, we analyze different notions of \emph{indistinguishability} security for \emph{quantum encryption of classical messages} (QECM) schemes, \emph{i.e.}~protocols that encrypt classical messages using quantum ciphertexts and classical keys. These notions can be formalized through multiplayer games involving pirating and receiving parties; we say that a 
scheme achieves indistinguishable security if the \emph{winning probability} of the following security game is negligible in the security parameter:
\begin{enumerate}[itemsep=1pt,parsep=1pt]
    \item The adversaries decide on the strategy and a pair of messages to distinguish $(m_0,m_1)$.
    \item The referee samples a uniformly random bit $b$ and sends the encryption of $m_b$ to the pirate adversary $\mc{A}$.
    \item $\mc{A}$ applies a CPTP map to the ciphertext states, getting a state in the receiving adversaries' registers. She sends these registers to the corresponding adversaries.
    \item The receiving adversaries are given the encryption key by the referee and attempt to guess the bit $b$, without communicating.
    \item The adversaries win if all  receivers are able to guess $b$ correctly.
\end{enumerate}
The difference between \emph{uncloneable-indistinguishable} security (for UE) and \emph{untelegraphable-indistinguishable} security (for UTE) lies in the nature of the CPTP map: in UE it may be arbitrary, while in UTE it must be quantum-to-classical (in which case the output can be freely copied). Also, for UTE, the security game is played with one receiver rather than two or more as for UE, but since classical information can be perfectly cloned, the optimal winning probability is the same for any number of receivers. When $\advA$ receives $t$ copies of the ciphertext, we obtain the notions of \emph{$t$-copy uncloneable-indistinguishable} security and \emph{$t$-copy untelegraphable-indistinguishable} security, respectively.

In this work, we prove that the Haar-measure scheme (\cref{def:Haar}), satisfies several forms of untelegraphable security.

\begin{definition}[Haar-measure encryption] \label{def:Haar}
     For a $\log n$-bit message $m \in [n]$ and a Haar-random unitary $U \in \mathcal{U}(d)$ as the key, encryption is given by
    \begin{equation*}
        \Enc(m,U) \coloneqq U \big( \,\underbracket{\,\ketbra{m}{m}\,{}}_{n \times n \text{ matrix}} \;\otimes\, \underbracket{\,I_{d/n}\,{}}_{\text{identity}} \big) U^*,
    \end{equation*}
    and decryption is performed by first applying $U^*$ to cancel the conjugation, and then measuring the first register.
\end{definition}

We also study efficient versions of this encryption arising from unitary designs and pseudorandom unitaries 
\ifExtAbs
\else
  (see \Cref{rem:efficient_construction})
\fi
. This yields efficient schemes which satisfy the security bounds of \cref{res:1,res:2,res:3,res:4} below.

\begin{result}
\ifExtAbs
\else
[\Cref{thm:haar_random_encryption_scheme_untelegraphable_indistinguishable}]
\fi
\label{res:1}
    The Haar-measure encryption scheme for classical bits ($n=2$ messages) achieves one-copy untelegraphable-indistinguishable security with winning probability upper bounded by
\ifExtAbs
   $\tfrac{1}{2} + \tfrac{1}{2 \sqrt{d+1}}$.
\else
    \begin{equation*}
        \tfrac{1}{2} + \tfrac{1}{2 \sqrt{d+1}}.
    \end{equation*}
\fi
\end{result}

\begin{result}
\ifExtAbs
\else
[\Cref{thm:haar_random_encryption_scheme_t_copy_untelegraphable_indistinguishable}]
\fi
\label{res:2}
    The Haar-measure encryption scheme for $n$ messages achieves $t$-copy untelegraphable-indistinguishable security with winning probability upper bounded by
    \ifExtAbs
        $\tfrac{1}{2} + \tfrac{7t\sqrt{n}}{\sqrt{d}}.$
    \else
    \begin{equation*}
        \tfrac{1}{2} + \tfrac{7t\sqrt{n}}{\sqrt{d}}.
    \end{equation*}
    \fi
\end{result}

We also consider the notion of collusion-resistant security, introduced by \cite{CKNY24arxiv}, in which $\advA$ may adaptively select multiple message pairs across $Q$ successive rounds, and its everlasting variant where adversaries are computationally bounded.

\begin{result}
\ifExtAbs
\else
[\Cref{thm:collusion}]
\fi
\label{res:3}
    The Haar-measure encryption scheme for $n$ messages achieves $Q$-round collusion-resistant untelegraphable-indistinguishable security with winning probability upper bounded by
    \ifExtAbs
        $\tfrac{1}{2} + \tfrac{7Q\sqrt{n}}{\sqrt{d}}.$
    \else
        \begin{equation*}
            \tfrac{1}{2} + \tfrac{7Q\sqrt{n}}{\sqrt{d}}.
        \end{equation*}
    \fi
\end{result}

\begin{result}
\ifExtAbs
\else
[\Cref{cor:collusion}]
\fi
\label{res:4}
    The Haar-measure encryption scheme achieves unconditional collusion-resistant untelegraphable-indistinguishable security for a polynomially-bounded number of rounds; and everlasting security for an unbounded number of rounds, under the assumption of pseudorandom unitaries.
\end{result}

This improves upon the result of \cite{CKNY24arxiv}, where collusion-resistant security relied on pseudorandom functions (an assumption stronger than pseudorandom unitaries), bounded-round security relied on pseudorandom states, and everlasting security relied on the QROM. Our plain-model achievability of everlasting security in the private key model contrasts with the result of~\cite{CKNY24arxiv} showing impossibility in the public-key~model.\looseness=-1

Then, in analogy with the informational equivalence of no-cloning and no-telegraphing principles in specific asymptotic regimes, we show that UTE can be expressed as a limit of UE as the number of receiving adversaries increases.

\begin{result}
\ifExtAbs
\else
[\Cref{thm:UTE_UE_asymptotic_equivalence}]
\fi
    For any UE scheme, the winning probability of the uncloneability-indistinguishability game with $s$ receiving adversaries converges to that of the untelegraphability-indistinguishability game at rate
    \ifExtAbs
        $\mathcal{O} \big( \tfrac{1}{\sqrt[3]{s}} \big).$
    \else
        \begin{equation*}
            \mathcal{O} \big( \tfrac{1}{\sqrt[3]{s}} \big).
        \end{equation*}
    \fi
\end{result}

Additionally we derive several lower bounds for the securities of the Haar-measure scheme. Since the winning probability for UE indistinguishability is always at least that for UTE, any lower bound for UTE also applies to UE
\ifExtAbs
  (see \Cref{fig:haar-bounds}). 
\else
(see \Cref{fig:telegraphing}). 
\fi

\ifExtAbs
\begin{figure}[t]
  \centering
  \begin{subfigure}[t]{0.48\textwidth}
    \centering
    \includegraphics[width=\linewidth]{section/Fig_1_compressed.jpg}
    \caption{Telegraphing value bounds.
    \label{fig1-a}}
  \end{subfigure}
  \hfill
  \begin{subfigure}[t]{0.48\textwidth}
    \centering
    \includegraphics[width=\linewidth]{section/Fig_2_compressed.jpg}
    \caption{Cloning value bounds.
    \label{fig1-b}}
  \end{subfigure}
  \caption{Bounds on the Haar-measure encryption of one bit. The outlined white regions show the range of possible values. See the full version for higher-resolution images and detailed captions. \cref{res:1} provides the upper bound in \cref{fig1-a} and \cref{res:6} provides the lower bounds in \cref{fig1-a,fig1-b}.}
  \label{fig:haar-bounds}
\end{figure}
 
\else
  \begin{figure}
\centering
\begin{tikzpicture}
        \draw[gray] (0.0,0) node[below]{{\color{black}$0$}} -- (0.0,8) (1.5,0) node[below]{{\color{black}$8$}} -- (1.5,8) (3.0,0) node[below]{{\color{black}$16$}} -- (3.0,8) (4.5,0) node[below]{{\color{black}$24$}} -- (4.5,8) (6.0,0) node[below]{{\color{black}$32$}} -- (6.0,8) (7.5,0) node[below]{{\color{black}$40$}} -- (7.5,8) (9.0,0) node[below]{{\color{black}$48$}} -- (9.0,8) (10.5,0) node[below]{{\color{black}$56$}} -- (10.5,8) (12.0,0) node[below]{{\color{black}$64$}} -- (12.0,8) (13.5,0) node[below]{{\color{black}$72$}} -- (13.5,8) (0,0.0) node[left]{{\color{black}$0.4$}} -- (15,0.0) (0,1.3333333333333333) node[left]{{\color{black}$0.5$}} -- (15,1.3333333333333333) (0,2.6666666666666665) node[left]{{\color{black}$0.6$}} -- (15,2.6666666666666665) (0,4.0) node[left]{{\color{black}$0.7$}} -- (15,4.0) (0,5.333333333333333) node[left]{{\color{black}$0.8$}} -- (15,5.333333333333333) (0,6.666666666666667) node[left]{{\color{black}$0.9$}} -- (15,6.666666666666667);
        \draw[Latex-Latex] (0,8) -- (0,0) -- (15,0);
        \node at (7.5,-0.8){Dimension};
        \node at (-1,4.0){\rotatebox{90}{Telegraphing value}};
        \fill[pattern=dots] (0.0,8.0) -- (0.375,4.666666666666667) -- (0.75,3.833333333333333) -- (1.125,3.4166666666666665) -- (1.5,3.15625) -- (1.875,2.973958333333333) -- (2.25,2.837239583333333) -- (2.625,2.729817708333333) -- (3.0,2.642537434895833) -- (3.375,2.5698038736979165) -- (3.75,2.5079803466796875) -- (4.125,2.4545873006184893) -- (4.5,2.4078683853149414) -- (4.875,2.3665401140848794) -- (5.25,2.329639871915181) -- (5.625,2.296429653962453) -- (6.0,2.266332893942793) -- (6.375,2.238891730395456) -- (6.75,2.2137373304770636) -- (7.125,2.190568804236439) -- (7.5,2.1691379174638614) -- (7.875,2.1492378083178965) -- (8.25,2.13069452479552) -- (8.625,2.11336058585069) -- (9.0,2.0971100180899116) -- (9.375,2.08183448439478) -- (9.75,2.067440231489752) -- (10.125,2.0538456593016705) -- (10.5,2.0409793677665222) -- (10.875,2.028778574069398) -- (11.25,2.01718782005713) -- (11.625,2.00615790898094) -- (12.0,1.9956450249864455) -- (12.375,1.9856099993553387) -- (12.75,1.97601769544325) -- (13.125,1.9668364902702504) -- (13.5,1.9580378353127945) -- (13.875,1.9495958825833408) -- (14.25,1.9414871648300513) -- (14.625,1.933690320836504) -- (15.0,1.9261858584927143) --(15,0) -- (0,0) -- cycle;
        \fill[pattern=crosshatch] (0.0,8.0) -- (0.375,5.1823351279308385) -- (0.75,4.314757303333052) -- (1.125,3.853096486728182) -- (1.5,3.555555555555555) -- (1.875,3.343408963851758) -- (2.25,3.1823339874174312) -- (2.625,3.054659264981073) -- (3.0,2.9502375002422205) -- (3.375,2.8627715591370784) -- (3.75,2.788119268239949) -- (4.125,2.7234294270471646) -- (4.5,2.666666666666666) -- (4.875,2.616333931532502) -- (5.25,2.5713022545136788) -- (5.625,2.5307020135118328) -- (6.0,2.493851039704652) -- (6.375,2.460205672971355) -- (6.75,2.4293265820357153) -- (7.125,2.400854358700581) -- (7.5,2.3744917459240407) -- (7.875,2.349990468884031) -- (8.25,2.3271413233332394) -- (8.625,2.305766609985963) -- (9.0,2.285714285714285) -- (9.375,2.266853389352006) -- (9.75,2.249070426324593) -- (10.125,2.232266483284322) -- (10.5,2.2163549047100295) -- (10.875,2.2012594065388256) -- (11.25,2.1869125328859726) -- (11.625,2.173254384464949) -- (12.0,2.1602315639280567) -- (12.375,2.1477962957087016) -- (12.75,2.1359056872384605) -- (13.125,2.124521105462569) -- (13.5,2.1136076479742036) -- (13.875,2.1031336922528343) -- (14.25,2.093070509730919) -- (14.625,2.0833919339506832) -- (15.0,2.0740740740740744) -- (15,8) -- (0,8) -- cycle;
        \draw[thick] (15.0,1.9261858584927143) -- (14.625,1.933690320836504) -- (14.25,1.9414871648300513) -- (13.875,1.9495958825833408) -- (13.5,1.9580378353127945) -- (13.125,1.9668364902702504) -- (12.75,1.97601769544325) -- (12.375,1.9856099993553387) -- (12.0,1.9956450249864455) -- (11.625,2.00615790898094) -- (11.25,2.01718782005713) -- (10.875,2.028778574069398) -- (10.5,2.0409793677665222) -- (10.125,2.0538456593016705) -- (9.75,2.067440231489752) -- (9.375,2.08183448439478) -- (9.0,2.0971100180899116) -- (8.625,2.11336058585069) -- (8.25,2.13069452479552) -- (7.875,2.1492378083178965) -- (7.5,2.1691379174638614) -- (7.125,2.190568804236439) -- (6.75,2.2137373304770636) -- (6.375,2.238891730395456) -- (6.0,2.266332893942793) -- (5.625,2.296429653962453) -- (5.25,2.329639871915181) -- (4.875,2.3665401140848794) -- (4.5,2.4078683853149414) -- (4.125,2.4545873006184893) -- (3.75,2.5079803466796875) -- (3.375,2.5698038736979165) -- (3.0,2.642537434895833) -- (2.625,2.729817708333333) -- (2.25,2.837239583333333) -- (1.875,2.973958333333333) -- (1.5,3.15625) -- (1.125,3.4166666666666665) -- (0.75,3.833333333333333) -- (0.375,4.666666666666667) -- (0.0,8.0) -- (0.375,5.1823351279308385) -- (0.75,4.314757303333052) -- (1.125,3.853096486728182) -- (1.5,3.555555555555555) -- (1.875,3.343408963851758) -- (2.25,3.1823339874174312) -- (2.625,3.054659264981073) -- (3.0,2.9502375002422205) -- (3.375,2.8627715591370784) -- (3.75,2.788119268239949) -- (4.125,2.7234294270471646) -- (4.5,2.666666666666666) -- (4.875,2.616333931532502) -- (5.25,2.5713022545136788) -- (5.625,2.5307020135118328) -- (6.0,2.493851039704652) -- (6.375,2.460205672971355) -- (6.75,2.4293265820357153) -- (7.125,2.400854358700581) -- (7.5,2.3744917459240407) -- (7.875,2.349990468884031) -- (8.25,2.3271413233332394) -- (8.625,2.305766609985963) -- (9.0,2.285714285714285) -- (9.375,2.266853389352006) -- (9.75,2.249070426324593) -- (10.125,2.232266483284322) -- (10.5,2.2163549047100295) -- (10.875,2.2012594065388256) -- (11.25,2.1869125328859726) -- (11.625,2.173254384464949) -- (12.0,2.1602315639280567) -- (12.375,2.1477962957087016) -- (12.75,2.1359056872384605) -- (13.125,2.124521105462569) -- (13.5,2.1136076479742036) -- (13.875,2.1031336922528343) -- (14.25,2.093070509730919) -- (14.625,2.0833919339506832) -- (15.0,2.0740740740740744);
\end{tikzpicture}
\caption{Bounds on the telegraphing value of the Haar-measure encryption of one bit, where the outlined white region is the range of possible values. The crosshatched~\begin{tikzpicture}\fill[pattern=crosshatch] (0,0) rectangle (1em,1em);\end{tikzpicture} region represents the upper bound of \cref{thm:haar_random_encryption_scheme_untelegraphable_indistinguishable} and the dotted~\begin{tikzpicture}\fill[pattern=dots] (0,0) rectangle (1em,1em);\end{tikzpicture} region represents the general lower bound of \cref{prop:lowerbound}. \label{fig:telegraphing}}
\end{figure}
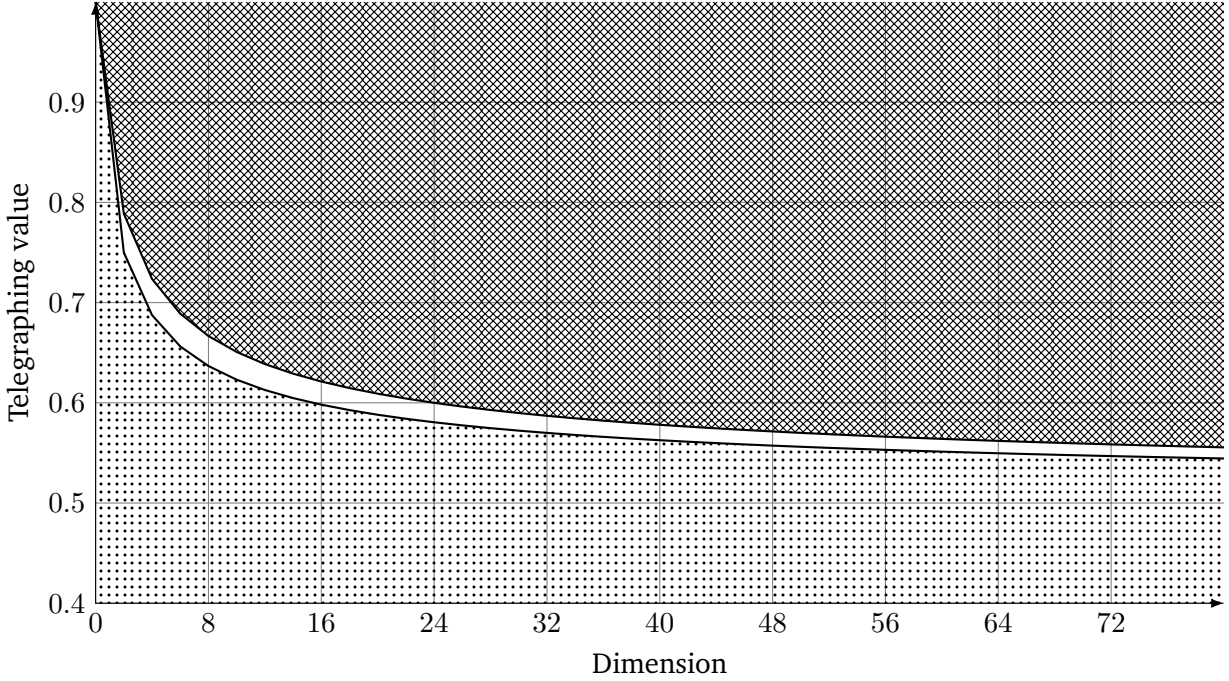
\fi

\begin{result}\label{res:6}
\ifExtAbs
\else
[\Cref{cor:better-lowerbound}]
\fi
    The Haar-measure encryption scheme for classical bits ($n=2$ messages) achieves one-copy untelegraphable-indistinguishable security with winning probability lower bounded by
    \ifExtAbs
       $\tfrac{1}{2} + \tfrac{1}{\sqrt{2\pi d}} + \mathcal{O} \big( \tfrac{1}{d^{3/2}} \big)$.
    \else
        \begin{equation*}
            \tfrac{1}{2} + \tfrac{1}{\sqrt{2\pi d}} + \mathcal{O} \big( \tfrac{1}{d^{3/2}} \big).
        \end{equation*}
    \fi
\end{result}

\begin{result}
\ifExtAbs
\else
[\Cref{cor:Haar_measure_scheme_lower_bound}]
\fi
    The Haar-measure encryption scheme for $n$ messages achieves $t$-copy untelegraphable-indistinguishable security with winning probability lower bounded by
    \ifExtAbs
       $\tfrac{1}{2} + \tfrac{\sqrt{tn}}{6\sqrt{\pi d}} + \mathcal{O} \big( \tfrac{\sqrt{t}\sqrt[3]{n}}{\sqrt[3]{d}} \big)$.
    \else
        \begin{equation*}
            \tfrac{1}{2} + \tfrac{\sqrt{tn}}{6\sqrt{\pi d}} + \mathcal{O} \big( \tfrac{\sqrt{t}\sqrt[3]{n}}{\sqrt[3]{d}} \big).
        \end{equation*}
    \fi
\end{result}

\ifExtAbs
\else
\begin{figure}[ht]
\centering
\begin{tikzpicture}
        \draw[gray] (0.0,0) node[below]{{\color{black}$0$}} -- (0.0,8) (1.875,0) node[below]{{\color{black}$5$}} -- (1.875,8) (3.75,0) node[below]{{\color{black}$10$}} -- (3.75,8) (5.625,0) node[below]{{\color{black}$15$}} -- (5.625,8) (7.5,0) node[below]{{\color{black}$20$}} -- (7.5,8) (9.375,0) node[below]{{\color{black}$25$}} -- (9.375,8) (11.25,0) node[below]{{\color{black}$30$}} -- (11.25,8) (13.125,0) node[below]{{\color{black}$35$}} -- (13.125,8) (0,0.0) node[left]{{\color{black}$0.4$}} -- (15,0.0) (0,1.3333333333333333) node[left]{{\color{black}$0.5$}} -- (15,1.3333333333333333) (0,2.6666666666666665) node[left]{{\color{black}$0.6$}} -- (15,2.6666666666666665) (0,4.0) node[left]{{\color{black}$0.7$}} -- (15,4.0) (0,5.333333333333333) node[left]{{\color{black}$0.8$}} -- (15,5.333333333333333) (0,6.666666666666667) node[left]{{\color{black}$0.9$}} -- (15,6.666666666666667);
        \draw[Latex-Latex] (0,8) -- (0,0) -- (15,0);
        \node at (7.5,-0.8){Number of qubits};
        \node at (-1,4.0){\rotatebox{90}{Cloning value}};
        \fill[pattern=crosshatch] (3.75,7.977189523108058) -- (4.125,7.623209003582964) -- (4.5,7.30827083453526) -- (4.875,7.02631751508886) -- (5.25,6.772411793415625) -- (5.625,6.542520794144692) -- (6.0,6.333333333333333) -- (6.375,6.1421131465690255) -- (6.75,5.966583334935902) -- (7.125,5.804835979063423) -- (7.5,5.6552614282206966) -- (7.875,5.51649278359882) -- (8.25,5.387362077549058) -- (8.625,5.266865469035085) -- (9.0,5.15413541726763) -- (9.375,5.0484182851531125) -- (9.75,4.949056193441867) -- (10.125,4.855472223824792) -- (10.5,4.767158277660193) -- (10.875,4.6836650541109695) -- (11.25,4.604593730405679) -- (11.625,4.529589017453897) -- (12.0,4.458333333333333) -- (12.375,4.390541890520274) -- (12.75,4.325958534068827) -- (13.125,4.264352200159029) -- (13.5,4.205513889690173) -- (13.875,4.149254071511144) -- (14.25,4.095400445672062) -- (14.625,4.043796009672948) -- (15.0,3.9942973807770152) -- (15,8) -- (4.125,8) -- cycle;
        \fill[pattern=dots] (0.375,2.1666666666666665) -- (0.75,1.7499999999999998) -- (1.125,1.5416666666666665) -- (1.5,1.4374999999999998) -- (1.875,1.3854166666666665) -- (2.25,1.3593749999999998) -- (2.625,1.3463541666666665) -- (3.0,1.3398437499999998) -- (3.375,1.3365885416666665) -- (3.75,1.3349609374999998) -- (4.125,1.3341471354166665) -- (4.5,1.3337402343749998) -- (4.875,1.3335367838541665) -- (5.25,1.3334350585937498) -- (5.625,1.3333841959635415) -- (6.0,1.3333587646484373) -- (6.375,1.3333460489908853) -- (6.75,1.3333396911621092) -- (7.125,1.3333365122477212) -- (7.5,1.3333349227905271) -- (7.875,1.3333341280619302) -- (8.25,1.3333337306976316) -- (8.625,1.3333335320154824) -- (9.0,1.3333334326744077) -- (9.375,1.3333333830038705) -- (9.75,1.3333333581686018) -- (10.125,1.3333333457509675) -- (10.5,1.3333333395421503) -- (10.875,1.3333333364377418) -- (11.25,1.3333333348855374) -- (11.625,1.3333333341094353) -- (12.0,1.3333333337213842) -- (12.375,1.3333333335273587) -- (12.75,1.3333333334303459) -- (13.125,1.3333333333818396) -- (13.5,1.3333333333575863) -- (13.875,1.3333333333454598) -- (14.25,1.3333333333393964) -- (14.625,1.3333333333363648) -- (15.0,1.333333333334849) -- (15,0) -- (0.375,0) -- cycle;
        \fill[pattern=checkerboard] (0.375,2.1666666666666665) -- (0.75,1.7499999999999998) -- (1.125,1.5416666666666665) -- (1.5,1.4374999999999998) -- (1.875,1.3854166666666665) -- (2.25,1.3593749999999998) -- (2.625,1.3463541666666665) -- (3.0,1.3398437499999998) -- (3.375,1.3365885416666665) -- (3.75,1.3349609374999998) -- (4.125,1.3341471354166665) -- (4.5,1.3337402343749998) -- (4.875,1.3335367838541665) -- (5.25,1.3334350585937498) -- (5.625,1.3333841959635415) -- (6.0,1.3333587646484373) -- (6.375,1.3333460489908853) -- (6.75,1.3333396911621092) -- (7.125,1.3333365122477212) -- (7.5,1.3333349227905271) -- (7.875,1.3333341280619302) -- (8.25,1.3333337306976316) -- (8.625,1.3333335320154824) -- (9.0,1.3333334326744077) -- (9.375,1.3333333830038705) -- (9.75,1.3333333581686018) -- (10.125,1.3333333457509675) -- (10.5,1.3333333395421503) -- (10.875,1.3333333364377418) -- (11.25,1.3333333348855374) -- (11.625,1.3333333341094353) -- (12.0,1.3333333337213842) -- (12.375,1.3333333335273587) -- (12.75,1.3333333334303459) -- (13.125,1.3333333333818396) -- (13.5,1.3333333333575863) -- (13.875,1.3333333333454598) -- (14.25,1.3333333333393964) -- (14.625,1.3333333333363648) -- (15.0,1.333333333334849) -- (15.0,1.3333384061467544) -- (14.625,1.3333405073748716) -- (14.25,1.333343478960174) -- (13.875,1.3333476814164102) -- (13.5,1.3333536245870163) -- (13.125,1.333362029499489) -- (12.75,1.333373915840698) -- (12.375,1.3333907256656434) -- (12.0,1.3334144983480631) -- (11.625,1.3334481179979552) -- (11.25,1.333495663362793) -- (10.875,1.333562902662576) -- (10.5,1.333657993392253) -- (10.125,1.3337924719918186) -- (9.75,1.3339826534511743) -- (9.375,1.3342516106503055) -- (9.0,1.3346319735690153) -- (8.625,1.335169887967278) -- (8.25,1.335930613804696) -- (7.875,1.3370064426012216) -- (7.5,1.3385278942760603) -- (7.125,1.3406795518691115) -- (6.75,1.3437224552187876) -- (6.375,1.3480257704048897) -- (6.0,1.354111577104241) -- (5.625,1.3627182074764448) -- (5.25,1.3748898208751485) -- (4.875,1.3921030816195579) -- (4.5,1.4164463084169654) -- (4.125,1.4508728299057827) -- (3.75,1.4995592835005964) -- (3.375,1.568412326478231) -- (3.0,1.6657852336678611) -- (2.625,1.8034913196231304) -- (2.25,1.9982371340023874) -- (1.875,2.2736493059129277) -- (1.5,2.642537434895833) -- (1.125,3.15625) -- (0.75,3.833333333333333) -- (0.375,4.666666666666667) -- cycle;
        \draw[thick] (15.0,3.9942973807770152) -- (14.625,4.043796009672948) -- (14.25,4.095400445672062) -- (13.875,4.149254071511144) -- (13.5,4.205513889690173) -- (13.125,4.264352200159029) -- (12.75,4.325958534068827) -- (12.375,4.390541890520274) -- (12.0,4.458333333333333) -- (11.625,4.529589017453897) -- (11.25,4.604593730405679) -- (10.875,4.6836650541109695) -- (10.5,4.767158277660193) -- (10.125,4.855472223824792) -- (9.75,4.949056193441867) -- (9.375,5.0484182851531125) -- (9.0,5.15413541726763) -- (8.625,5.266865469035085) -- (8.25,5.387362077549058) -- (7.875,5.51649278359882) -- (7.5,5.6552614282206966) -- (7.125,5.804835979063423) -- (6.75,5.966583334935902) -- (6.375,6.1421131465690255) -- (6.0,6.333333333333333) -- (5.625,6.542520794144692) -- (5.25,6.772411793415625) -- (4.875,7.02631751508886) -- (4.5,7.30827083453526) -- (4.125,7.623209003582964) -- (3.75,7.977189523108058) --(0.375,8) -- (0.375,4.666666666666667) -- (0.75,3.833333333333333) -- (1.125,3.15625) -- (1.5,2.642537434895833) -- (1.875,2.2736493059129277) -- (2.25,1.9982371340023874) -- (2.625,1.8034913196231304) -- (3.0,1.6657852336678611) -- (3.375,1.568412326478231) -- (3.75,1.4995592835005964) -- (4.125,1.4508728299057827) -- (4.5,1.4164463084169654) -- (4.875,1.3921030816195579) -- (5.25,1.3748898208751485) -- (5.625,1.3627182074764448) -- (6.0,1.354111577104241) -- (6.375,1.3480257704048897) -- (6.75,1.3437224552187876) -- (7.125,1.3406795518691115) -- (7.5,1.3385278942760603) -- (7.875,1.3370064426012216) -- (8.25,1.335930613804696) -- (8.625,1.335169887967278) -- (9.0,1.3346319735690153) -- (9.375,1.3342516106503055) -- (9.75,1.3339826534511743) -- (10.125,1.3337924719918186) -- (10.5,1.333657993392253) -- (10.875,1.333562902662576) -- (11.25,1.333495663362793) -- (11.625,1.3334481179979552) -- (12.0,1.3334144983480631) -- (12.375,1.3333907256656434) -- (12.75,1.333373915840698) -- (13.125,1.333362029499489) -- (13.5,1.3333536245870163) -- (13.875,1.3333476814164102) -- (14.25,1.333343478960174) -- (14.625,1.3333405073748716) -- (15.0,1.3333384061467544);
\end{tikzpicture}
\caption{Bounds on the cloning value of the Haar-measure encryption of one bit, where the outlined white region is the range of possible values. The crosshatched~\begin{tikzpicture}\fill[pattern=crosshatch] (0,0) rectangle (1em,1em);\end{tikzpicture} region represents the upper bound due to~\cite{BC25arXiv}, the dotted~\begin{tikzpicture}\fill[pattern=dots] (0,0) rectangle (1em,1em);\end{tikzpicture} region represents the lower bound due to~\cite{MST21arxiv}, and the checkered~\begin{tikzpicture}\fill[pattern=checkerboard] (0,0) rectangle (1em,1em);\end{tikzpicture} region represents the improved lower bound of \cref{prop:lowerbound}.
\label{fig:bounds}}
\end{figure}
\fi

In particular our result on the one-copy untelegraphable-indistinguishable security for classical bits, of the Haar-measure scheme, is tight (to order).

Finally, we prove a minimality property of the Haar-measure scheme extending the result of~\cite{MST21arxiv} that yields general lower bounds for UTE, and consequently for UE 
\ifExtAbs
  (see \Cref{fig:haar-bounds}). 
\else
(see \Cref{fig:bounds}). 
\fi

\begin{result}
\ifExtAbs
\else
[\Cref{cor:UTE_and_UE_lower_bounds}]
\fi
    For any UE scheme with ciphertext dimension $d$, the winning probability of the uncloneability-indistinguishability game is lower bounded by
    \ifExtAbs
       $\tfrac{1}{2} + \Omega \big( \tfrac{1}{\sqrt{d}} \big)$.
    \else
        \begin{equation*}
            \tfrac{1}{2} + \Omega \big( \tfrac{1}{\sqrt{d}} \big).
        \end{equation*}
    \fi
\end{result}

This lower bound improves upon the previously best-known bound for UE, namely $\tfrac{1}{2} + \Omega \big( \frac{1}{d} \big)$, established in~\cite{MST21arxiv}.

We further study a strengthened notion of UTE security where adversarial CPTP maps are restricted only to entanglement-breaking channels, which strictly generalize quantum-to-classical maps. This yields novel generalised guarantees for both untelegraphable-indistinguishable and $t$-copy untelegraphable-indistinguishable security 
\ifExtAbs
  .
\else
  (see~\cref{lem:tg-char}).
\fi

\ifExtAbs
\else
 \subsection{Open questions}
 \begin{enumerate}
    \item Untelegraphable encryption constitutes a restriction of uncloneable encryption in the sense that it imposes a strict classicality on the messages that the pirate can send to the receiving adversaries. There are also a variety of intermediate restrictions that can be imposed on the pirate channel, \textit{e.g.}, bounded storage~\cite{DFSS05} or noisy communication~\cite{WST08}. These do not achieve the full generality of cloning attacks, but they do give rise to a much wider range of possible attacks than telegraphing. Is it possible to find similar upper bounds on the value to the ones we find in those models?
    
    \item In this work, we use untelegraphable encryption to constrain the value of uncloneable encryption. Is it possible to use this line of reasoning to find stronger results? For example, could we show the existence of an uncloneable bit by reducing to an untelegraphable encryption protocol? 
\end{enumerate}
\fi

\ifExtAbs
\else

\subsection{Organization}

The remainder of this paper is organized as follows.
\Cref{sec:preliminaries} introduces notation, background in quantum information theory, and formal definitions of the encryption schemes and security notions.
\Cref{sec:haar_random} defines the Haar-random scheme and proves untelegraphable-indistinguishability security.
\Cref{sec:collusion} proves that the Haar-measure scheme is secure against telegraphing-distinguishing attacks with collusion.
\Cref{sec:UTE_as_UE_limit} establishes the asymptotic equivalence of UTE and UE.
\Cref{sec:lowerbounds} derives lower bounds for Haar-measure security.
\Cref{sec:minimality} proves a minimality property of the Haar-measure scheme and lower bounds for all UTE and UE schemes.

\fi

\else
  \begin{abstract}
    We investigate the notion of untelegraphable encryption (UTE), a quantum encryption primitive that is a special case of uncloneable encryption (UE), where the adversary’s capabilities are restricted to producing purely classical information rather than arbitrary quantum states.
    We present an unconditionally secure construction of UTE that achieves untelegraphable-indistinguishability security, together with natural multi-ciphertext and bounded collusion-resistant extensions, without requiring any additional assumptions. We also extend this to the unbounded case, assuming pseudo-random unitaries, yielding everlasting security. Furthermore, 
    we derive results on UE using approaches from UTE in the following ways: first, we provide new lower bounds on UTE, which give new lower bounds on UE; second, we prove an asymptotic equivalence between UTE and UE in the regime where the number of adversaries in UE grows.
    These results suggest that UTE may provide a new path toward achieving a central open problem in the area: indistinguishability security for UE in the plain model.
\end{abstract}
\newpage
\tableofcontents
  \newpage

\section{Preliminaries} \label{sec:preliminaries}

We let $\NN$ denote the set of natural numbers, and write $[n]$ for the set $\{0, \dots, n-1\}$. The logarithm with base $b$ is denoted by $\log_b(\cdot)$, and in particular, we use $\log(\cdot)$ to denote the binary logarithm.

All Hilbert spaces considered in this work are assumed to be finite-dimensional. Given a finite set $A$, write $H_A=\C^A$ for the Hilbert space with canonical orthonormal basis $\set*{\ket{a}}{a\in A}$; in this case, $A$ is called a register. For Hilbert spaces $H$ and $K$, we denote by $\mathcal{B}(H, K)$ the space of bounded linear operators from $H$ to $K$. In the special case $K = H$, we simply write $\mathcal{B}(H)$. The identity operator on $H \simeq \CC^d$ is denoted by $I_d$, where the subscript is dropped if clear from context, and the trace over $H$ is written as $\Tr[\cdot]$. An operator $M \in \mathcal{B}(H)$ is positive semi-definite if and only if it is Hermitian (\ie, $M^* = M$) with non-negative eigenvalues, in which case we write $M \succeq 0$. The Hilbert space norm is denoted by $\norm{\cdot}$. The operator norm on $\mc{B}(H)$ is also denoted $\norm{\cdot}$, and the trace norm is denoted $\norm{\cdot}_{\Tr}=\frac{1}{2}\norm{\cdot}_1$.

A Hilbert space $H$ will be regarded as a quantum system. A quantum state (or density operator) on $H$ is a positive semi-definite operator $\rho \in \mathcal{B}(H)$ with unit trace. The set of all density operators is convex, with extreme points given by rank-one projectors $\ketbra{\psi}{\psi}$, where $\ket{\psi} \in H$ is a unit vector, referred to as a pure state. We shall use both vector and matrix notations for pure states interchangeably. We sometimes write $\rho_{A_1\cdots A_n}$ to mean $\rho\in\mc{B}(H_{A_1} \otimes \cdots \otimes H_{A_n})$.

For a composite system $H \simeq H_{A_1} \otimes \cdots \otimes H_{A_n}$ consisting of $n$ subsystems, the partial trace over the $i$-the subsystem is denoted by $\Tr_{A_i}[\cdot]$. We write partial traces of states $\rho_{A_1\cdots A_n}$ as $\rho_{A_1\cdots A{i-1}A_{i+1}\cdots A_n}=\Tr_{A_i}(\rho)$. When subsystems have different dimensions, we also index the partial trace by the dimension being traced out. For instance, if $H \simeq \CC^d \otimes \CC^D$, the notation $\Tr_D[\cdot]$ refers to the partial trace over the second tensor factor.

A linear map $\Phi: \mathcal{B}(H) \to \mathcal{B}(H)$ is said to be completely positive if, for every $k \in \NN$, the extended map $\Phi \otimes \id_k$ is positive (\ie, it maps positive operators to positive operators), where $\id_k$ denotes the identity map on $\mathcal{B}(\CC^k)$. The map $\Phi$ is trace preserving if $\Tr[\Phi(\rho)] = \Tr[\rho]$ for all $\rho \in \mathcal{B}(H)$. A quantum channel is a map that is both Completely Positive and Trace Preserving (\CPTP).

A positive operator-valued measure (\POVM) on $H$ is a finite family of positive semi-definite operators $\{M_i\}$ satisfying $\sum_i M_i = I$. A projection-valued measure (\PVM) is a special case of a \POVM in which each $M_i$ is an orthogonal projector. A generalised measurement on $H$ is a finite family of linear maps $\{V_i:H\rightarrow K\}$ for some Hilbert space $K$ such that $\sum_iV_i^\ast V_i=I$.

We denote by $\mathfrak{S}_n$ the symmetric group on $n$ elements. A unitary representation of a finite group $G$ is a group homomorphism $G\rightarrow\mc{U}(H)$. We make use of the representation $V_d:\mfk{S}_n\rightarrow\mc{U}((\C^d)^{\otimes n})$ defined via $V_d(\pi)(\ket{\psi_1}\otimes\cdots\otimes\ket{\psi_n})=\ket{\psi_{\pi^{-1}(1)}}\otimes\cdots\otimes\ket{\psi_{\pi^{-1}(n)}}$. The group of $d \times d$ unitary matrices is written as $\mathcal{U}(d)=\mc{U}(\C^d)$, consisting of all matrices $U$ such that $UU^* = U^*U = I_d$.

There is a unique translation-invariant probability measure on $\mc{U}(d)$ called the Haar measure. We denote this measure as $\mu_{Haar}$, since $d$ is usually clear from context. We denote integration with respect to the Haar measure by $\int\cdot\,dU=\int\cdot\,d\mu_{Haar}(U)$.

A family of functions $G_\lambda:K_\lambda\rightarrow\mc{U}(2^\lambda)$ is a pseudorandom unitary there is a quantum polynomial-time (in $\lambda$) algorithm that implements $k\mapsto G_{\lambda}(k)$, and any quantum polynomial-time distinguisher can only distinguish Haar-random $U$ from $G_\lambda(k)$ with uniformly random $k$ with negligible advantage, given oracle and inverse oracle access to either.

\subsection{Untelegraphable and uncloneable encryption}
\label{sec:definitions}
\emph{Untelegraphable encryption}, introduced in \cite{CKNY24arxiv}, is a relaxation of  \emph{uncloneable encryption} \cite{BL20}: a symmetric-key encryption scheme of classical messages in which ciphertexts are quantum states, designed to prevent unauthorized replication of information by adversaries.

\begin{definition}
    A \emph{quantum encryption of classical messages (QECM)} is a tuple $\ttt{Q}=(M,K,\pi,H,\{\sigma^k_m\}_{k\in K,m\in M})$, where
    \begin{itemize}
        \item $M$ is a finite set, called the set of \emph{messages};
        \item $K$ is a measureable space, called the set of \emph{keys};
        \item $\pi$ is a probability measure on $K$, called the \emph{key distribution};
        \item $H$ is a finite-dimensional Hilbert space, called the \emph{codespace};
        \item $\sigma^k_m$ are quantum states on $H$, called the \emph{ciphertexts}.
    \end{itemize}
    
    The QECM $\ttt{Q}$ is \emph{$\varepsilon$-correct} if for each $k\in K$, there is a POVM $\{P^k_m\}_{m\in M}$ such that for any $m\in M$,
    \begin{align*}
        \int_K\Tr(P^k_m\sigma^k_m)d\pi(k)\geq 1-\varepsilon.
    \end{align*}
    We say $\ttt{Q}$ is \emph{correct} if it is $0$-correct.

    An \emph{efficient QECM} is a collection of QECMs $\{\ttt{Q}_\lambda\}_{\lambda\in\N}$ such that there exists a negligible function $\eta$ such that $\ttt{Q}_\lambda$ is $\eta(\lambda)$-correct and a triple of quantum algorithms:
    \begin{itemize}
        \item $\Gen:\{1\}^\ast\rightarrow\bigcup_\lambda K_\lambda$ such that $\Gen(1^\lambda)$ samples from $\pi_\lambda$, called the \emph{key generation algorithm};
        \item $\Enc:\bigcup_\lambda\{1^\lambda\}\times M_\lambda\times K_\lambda\rightarrow\bigcup_\lambda D(H_\lambda)$ such that $\Enc(1^\lambda,m,k)=\sigma^{\lambda,k}_m$, called the \emph{encryption algorithm};
        \item $\Dec:\bigcup_\lambda \{1^\lambda\}\times D(H_\lambda)\times K_\lambda\rightarrow\bigcup_\lambda M_\lambda$ such that for all $m\in M_\lambda$, $\Pr[\Dec(1^\lambda,\sigma^{\lambda,k}_m,k)=m|k\leftarrow\pi]\geq 1-\eta(\lambda)$.
    \end{itemize}
\end{definition}

We can capture cloning and telegraphing attacks in tandem with the following general form of attack against a QECM.

\begin{definition} \label{def:cloning_attack}
    Let $N,t,s\in\N$ and let $\scr{F}$ be a collection of quantum channels. A \emph{$t$-to-$s$ $N$-message cloning attack over $\scr{F}$} against a QECM $\ttt{Q}$ is a tuple $\ttt{A}=(M_0,\{B_i\}_{i\in[s]},\{P^{i,k}_{m}\}_{i\in[s],k\in K,m\in M_0},\Phi)$, where
    \begin{itemize}
        \item $M_0\subseteq M$ is a set of $N$ messages;
        \item $B_i$ is a finite-dimensional Hilbert space for each $i$;
        \item For each $i\in[s]$ and $k\in[k]$, $\{P^{i,k}_m\}_{m\in M_0}\subseteq\mc{B}(B_i)$ is a POVM;
        \item $\Phi:\mc{B}(H^{\otimes t})\rightarrow\mc{B}(B_0\otimes\cdots\otimes B_{s-1})$ is  a quantum channel such that $\Phi\in\scr{F}$.
    \end{itemize}
    The \emph{cloning probability} of $\ttt{A}$ against $\ttt{Q}$ is
    \begin{align*}
        \mfk{c}^N_{t\rightarrow s}(\ttt{Q}|\ttt{A})=\int_K\frac{1}{N}\sum_{m\in M_0}\Tr\squ{(P^{1,k}_m\otimes\cdots\otimes P^{s,k}_m)\Phi((\sigma^k_m)^{\otimes t})}d\pi(k).
    \end{align*}
    The \emph{$t$-to-$s$ $N$-message cloning value over $\scr{F}$} of $\ttt{Q}$ is $\mfk{c}^N_{t\rightarrow s}(\ttt{Q}|\scr{F})=\sup_{\ttt{A}}\mfk{c}^N_{t\rightarrow s}(\ttt{Q}|\ttt{A})$, where the supremum is over all $t$-to-$s$ $N$-message cloning attacks from $\scr{F}$; $N$ is omitted when $N=|M|$ and $\scr{F}$ is omitted when it is the set of all channels. If $N=|M|$, we omit $M_0$ in $\ttt{A}$, and if $s=1$, we omit $i$ in $B_i$ and $P^{i,k}_m$.
\end{definition}

It is easy to see that if $\scr{F}$ is closed under pre- and post-composition with partial traces, then $\mfk{c}_{t\rightarrow s}^N(\ttt{Q}|\scr{F})\leq\mfk{c}_{t'\rightarrow s'}^N(\ttt{Q}|\scr{F})$ for $t'\geq t$ and $s'\leq s$.

We recover the various values studied in the context of uncloneable and untelegraphable encryption as special cases of this definition.

\begin{definition}\label{def:specifying-uncloneability}
    Let $\ttt{Q}$ be a QECM. Let $\scr{M}$ be the set of all measurement channels.
    \begin{itemize}
        \item A \emph{cloning attack} against $\ttt{Q}$ is a $1$-to-$2$ $|M|$-message cloning attack. The \emph{cloning value} of $\ttt{Q}$ is $\mfk{c}_{1\rightarrow 2}(\ttt{Q})$. We say $\ttt{Q}$ is \emph{$\varepsilon$-uncloneable secure} if $\mfk{c}_{1\rightarrow 2}(\ttt{Q})\leq\frac{1}{|M|}+\varepsilon$. An efficient QECM $\{\ttt{Q}_\lambda\}_\lambda$ is \emph{$\eta$-uncloneable secure} if $\ttt{Q}_\lambda$ is $\eta(\lambda)$-uncloneable secure for each $\lambda$. We say $\{\ttt{Q}_\lambda\}_\lambda$ is \emph{weakly uncloneable secure} if $\lim_{\lambda\rightarrow\infty}\eta(\lambda)=0$ and \emph{strongly uncloneable secure} if $\eta$ is a negligible function.
        \item A \emph{cloning-distinguishing attack} against $\ttt{Q}$ is a $1$-to-$2$ $2$-message cloning attack. The \emph{cloning-distinguishing value} of $\ttt{Q}$ is $\mfk{c}^2_{1\rightarrow 2}(\ttt{Q})$. We say $\ttt{Q}$ is \emph{$\varepsilon$-uncloneable-indistinguishable secure} if $\mfk{c}_{1\rightarrow 2}^2(\ttt{Q})\leq\frac{1}{2}+\varepsilon$. An efficient QECM $\{\ttt{Q}_\lambda\}_\lambda$ is \emph{$\eta$-uncloneable-indistinguishable secure} if $\ttt{Q}_\lambda$ is $\eta(\lambda)$-uncloneable-indistinguishable secure for each $\lambda$. We say $\{\ttt{Q}_\lambda\}_\lambda$ is \emph{weakly uncloneable-indistinguishable secure} if $\lim_{\lambda\rightarrow\infty}\eta(\lambda)=0$ and \emph{strongly uncloneable-indistinguishable secure} if $\eta$ is a negligible function.
        \item A \emph{telegraphing attack} is a $1$-to-$1$ $|M|$-message cloning attack over $\scr{M}$. The \emph{telegraphing value} of $\ttt{Q}$ is $\mfk{c}_{1\rightarrow 1}(\ttt{Q}|\scr{M})$. We say $\ttt{Q}$ is \emph{$\varepsilon$-untelegraphable secure} if $\mfk{c}_{1\rightarrow 1}(\ttt{Q}|\scr{M})\leq\frac{1}{|M|}+\varepsilon$. An efficient QECM $\{\ttt{Q}_\lambda\}_\lambda$ is \emph{$\eta$-untelegraphable secure} if $\ttt{Q}_\lambda$ is $\eta(\lambda)$-untelegraphable secure for each $\lambda$. We say $\{\ttt{Q}_\lambda\}_\lambda$ is \emph{weakly untelegraphable secure} if $\lim_{\lambda\rightarrow\infty}\eta(\lambda)=0$ and \emph{strongly untelegraphable secure} if $\eta$ is a negligible function.
        \item A \emph{telegraphing-distinguishing attack} is a $1$-to-$1$ $2$-message cloning attack over $\scr{M}$. The \emph{telegraphing-distinguishing value} of $\ttt{Q}$ is $\mfk{c}^2_{1\rightarrow 1}(\ttt{Q}|\scr{M})$. We say $\ttt{Q}$ is \emph{$\varepsilon$-untelegraphable-indistinguishable secure} if $\mfk{c}_{1\rightarrow 1}^2(\ttt{Q}|\scr{M})\leq\frac{1}{2}+\varepsilon$. An efficient QECM $\{\ttt{Q}_\lambda\}_\lambda$ is \emph{$\eta$-untelegraphable-indistinguishable secure} if $\ttt{Q}_\lambda$ is $\eta(\lambda)$-untelegraphable-indistinguishable secure for each $\lambda$. We say $\{\ttt{Q}_\lambda\}_\lambda$ is \emph{weakly untelegraphable-indistinguishable secure} if $\lim_{\lambda\rightarrow\infty}\eta(\lambda)=0$ and \emph{strongly untelegraphable-indistinguishable secure} if $\eta$ is a negligible function.
    \end{itemize}
\end{definition}

Observe that $\mfk{c}^2_{t \rightarrow s}(\ttt{Q} | \scr{M}) = \mfk{c}^2_{t \rightarrow s'}(\ttt{Q} | \scr{M})$ for any $s$ and $s'$, since the outputs of measurement channels in $\scr{M}$ can be prepared into an arbitrary number of copies.

\begin{definition}
    A \emph{efficient attack} against an efficient QECM $\{\ttt{Q}_\lambda\}_\lambda$ is a collection of attacks $\{\ttt{A}_\lambda=(M_0^\lambda,\{B_i^\lambda\},\{P^{\lambda,i,k}_m\},\Phi^\lambda)\}_\lambda$ such that $\ttt{A}_\lambda$ is an attack against $\ttt{Q}_\lambda$, and the attacks can be implemented in polynomial time in $\lambda$.

    An efficient QECM $\{\ttt{Q}_\lambda\}_\lambda$ is \emph{uncloneable-indistinguishable secure against efficient adversaries} if for every efficient attack $\{\ttt{A}_\lambda\}_\lambda$ against $\{\ttt{Q}_\lambda\}_\lambda$ where $\ttt{A}_\lambda$ is a cloning-distinguishing attack, $\mfk{c}_{1\rightarrow 2}^2(\ttt{Q}_\lambda|\ttt{A}_\lambda)=\frac{1}{2}+\negl(\lambda)$. An efficient QECM $\{\ttt{Q}_\lambda\}_\lambda$ is \emph{untelegraphable-indistinguishable secure against efficient adversaries} if for every efficient attack $\{\ttt{A}_\lambda\}_\lambda$ against $\{\ttt{Q}_\lambda\}$ where $\ttt{A}_\lambda$ is a telegraphing-distinguishing attack, $\mfk{c}_{1\rightarrow 1}^2(\ttt{Q}_\lambda|\ttt{A}_\lambda)=\frac{1}{2}+\negl(\lambda)$.

    An efficient QECM $\{\ttt{Q}_\lambda\}_\lambda$ is \emph{everlasting uncloneable-indistinguishable secure against efficient adversaries} if for every collection of attacks $\{\ttt{A}_\lambda\}_\lambda$ where $\ttt{A}_\lambda$ is a cloning-distinguishing attack against $\ttt{Q}_\lambda$ and $\Phi^\lambda$ can be efficiently implemented (but not necessarily the measurements $P^{\lambda,i,k}_m$), $\mfk{c}_{1\rightarrow 2}^2(\ttt{Q}_\lambda|\ttt{A}_\lambda)=\frac{1}{2}+\negl(\lambda)$. An efficient QECM $\{\ttt{Q}_\lambda\}_\lambda$ is \emph{everlasting untelegraphable-indistinguishable secure against efficient adversaries} if for every collection of attacks $\{\ttt{A}_\lambda\}_\lambda$ where $\ttt{A}_\lambda$ is a telegraphing-distinguishing attack against $\ttt{Q}_\lambda$ and $\Phi^\lambda$ can be efficiently implemented, $\mfk{c}_{1\rightarrow 1}^2(\ttt{Q}_\lambda|\ttt{A}_\lambda)=\frac{1}{2}+\negl(\lambda)$.
\end{definition}

In~\cite{CKNY24arxiv}, a further security notion was also considered, which is not captured by the above framework.

\begin{definition}
    A \emph{telegraphing-distinguishing attack with collusion} $\ttt{A}$ against a QECM $\ttt{Q}=(M,K,\pi,H,\{\sigma^k_m\}_{k,m})$ consists of $Q\in\N$ rounds, a finite set $X$ of telegraphing messages, a probability distribution $p(\cdot|x,k)$ on $\{0,1\}$ for all $x\in X$ and $k\in K$, Hilbert spaces $H_i$ for $i\in[Q+1]$ where $H_0=\C$, generalised measurements $\{V^{i}_{m_0,m_1}:H_i\otimes H\rightarrow H_{i+1}\}_{(m_0,m_1)\in M^2,m_0\neq m_1}$ for $i\in[Q]$, and a POVM $\{P_x\}_{x\in X}\subseteq B(H_Q\otimes H)$. The \emph{value} of $\ttt{A}$ against $\ttt{Q}$ is
    \begin{align*}
        \mfk{t}(\ttt{Q}|\ttt{A})=\int_K\frac{1}{2}\hspace{-0.3cm}\sum_{\substack{ b\in\{0,1\} \\[0.25em] x\in X \\[0.25em] m_0^{(1)},m_1^{(1)},\ldots,\\m_0^{(Q)},m_1^{(Q)}}}\hspace{-0.3cm}p(b|x,k)\Tr\squ*{P_x\cdot V^{Q-1}_{m_0^{(Q)},m_1^{(Q)}}\cdots (V^0_{m^{(1)}_0, m^{(1)}_{1}}(V^0_{m^{(1)}_0, m^{(1)}_{1}})^\dag \otimes\sigma^k_{m^{(1)}_b})\cdots V^{Q-1}_{m^{(Q)}_0,m^{(Q)}_1}\otimes\sigma^{(Q)}_{m^{(Q)}_b}} d\pi(k)
    \end{align*}
    An efficient QECM $\{\ttt{Q}_\lambda\}_{\lambda}$ is \emph{collusion-resistant untelegraphable-indistinguishable secure} if for any efficient telegraphing-distinguishing attack with collusion $\{\ttt{A}_\lambda\}_\lambda$, $\mfk{t}(\ttt{Q}_\lambda|\ttt{A}_\lambda)\leq\frac{1}{2}+\negl(\lambda)$.

    An efficient QECM $\{\ttt{Q}_\lambda\}_{\lambda}$ is \emph{everlasting collusion-resistant untelegraphable-indistinguishable secure} if for any collection of telegraphing-distinguishing attacks with collusion $\{\ttt{A}_\lambda\}_\lambda$ such that the $V^{\lambda,i}_{m_0,m_1}$ and $P^\lambda_x$ are efficiently implemented, $\mfk{t}(\ttt{Q}_\lambda|\ttt{A}_\lambda)\leq\frac{1}{2}+\negl(\lambda)$.
\end{definition}

It follows directly from the above definitions that collusion-resistant untelegraphable-indistinguishable security is a stronger security notion than multi-copy untelegraphable-indistinguishable security, in the sense that $\mfk{t}(\ttt{Q})\geq\mfk{c}_{Q\rightarrow 1}^N(\ttt{Q}|\scr{M})$.

Telegraphing attacks admit an alternate characterisation owing to the fact that the message the telegrapher sends is classical.

\begin{lemma}\label{lem:tg-char}
    Let $\ttt{Q}=(M,K,\pi,H,\{\sigma^k_m\})$ be a QECM. If $\ttt{A}=(M_0,B,\{P^k_m\},\Phi)$ is a $1$-to-$1$ $N$-message cloning attack where $\Phi$ is an entanglement-breaking channel, then there exists a set $X$, probability distributions $p(\cdot|x,k)$ over $M_0$ for all $x\in X$ and $k\in K$, and a POVM $\{P_x\}_{x\in X}\subseteq\mc{B}(H)$ such that
    $$\mfk{c}_{1\rightarrow 1}^N(\ttt{Q}|\ttt{A})=\int_K\frac{1}{N}\sum_{m\in M_0}\sum_{x\in X}p(m|x,k)\Tr\squ{P_x\sigma^k_m}d\pi(k).$$
    Conversely, for any set $X$, subset $M_0\subseteq M$ of size $N$, probability distributions $p(\cdot|x,k)$ over $M_0$ for all $x\in X$ and $k\in K$, and POVM $\{P_x\}_{x\in X}\subseteq\mc{B}(H)$, there exists a $1$-to-$1$ $N$-message cloning attack $\ttt{A}=(M_0,B,\{P^k_m\},\Phi)$ such that $\Phi$ is a measurement channel and the same equality holds.
\end{lemma}

It follows that, if $\scr{E}$ is the set of entanglement-breaking channels, then $\mfk{c}_{1\rightarrow 1}^N(\ttt{Q}|\scr{E})=\mfk{c}_{1\rightarrow 1}^N(\ttt{Q}|\scr{M})$.

\begin{proof}
    Let $\ttt{A}=(M_0,B,\{P^k_m\},\Phi)$ be a $1$-to-$1$ $N$-message cloning attack against $\ttt{Q}$ where $\Phi$ is entanglement-breaking. Then, there exists a POVM $\{P_x\}_{x\in X}$ and states $\sigma_x$ such that $\Phi(\rho)=\sum_{x\in X}\Tr\squ{P_x\rho}\sigma_x$. Let $p(m|x,k)=\Tr\squ{P^k_m\sigma_x}$. Therefore, we have that
    \begin{align*}
        \mfk{c}_{1\rightarrow 1}^N(\ttt{Q}|\ttt{A})&=\int_K\frac{1}{N}\sum_{m\in M_0}\Tr\squ{P^k_m\Phi(\sigma^k_m)}d\pi(k)\\
        &=\int_K\frac{1}{N}\sum_{m\in M_0}\sum_{x\in X}\Tr\squ{P^k_m\sigma_x}\Tr\squ{P_x\sigma^k_m}d\pi(k)\\
        &=\int_K\frac{1}{N}\sum_{m\in M_0}\sum_{x\in X}p(m|x,k)\Tr\squ{P_x\sigma^k_m}d\pi(k).
    \end{align*}

    For the converse, let $B=\C^X$, $P^k_m=\sum_{x\in X}p(m|x,k)\ketbra{x}$, $\Phi(\rho)=\sum_{x\in X}\Tr\squ{P_x\rho}\ketbra{x}$. By construction, this is a $1$-to-$1$ $N$-message cloning attack over $\scr{M}$ against $\ttt{Q}$, and by the same argument, the same equality holds.
\end{proof}

Uncloneable-indistinguishable security provides a stronger guarantee for a QECM scheme than untelegraphable-indistinguishable security. In particular, if a QECM scheme achieves uncloneable-indistinguishable security, then the cloning probability $\mfk{c}^2_{1\rightarrow 2}(\ttt{Q})$ is negligible. Since $\mfk{c}^2_{1\rightarrow 1}(\ttt{Q} | \scr{M}) \leq \mfk{c}^2_{1\rightarrow 2}(\ttt{Q})$, then the scheme also satisfies untelegraphable-indistinguishable security.

However the two notions are not equivalent: there exist schemes that satisfy untelegraphable-indistinguishable security but not uncloneable-indistinguishable security. Consider a QECM scheme that is $2$-copy untelegraphable-indistinguishable secure for a collection of ciphertexts $\{\rho\}$. Then, the modified $1$-copy untelegraphable-indistinguishable secure scheme with ciphertexts of the form $\{ \rho \,\otimes\, \rho \}$ fails to achieve uncloneable-indistinguishable security, as an attack can trivially apply the identity channel to win with certainty (see \cite{CKNY24arxiv} for a separation between UTE and UE with unbounded polynomial number of adversaries, under the classical oracle model).

\section{The security of Haar-measure encryption} \label{sec:haar_random}

The Haar-random QECM scheme was introduced by \cite{MST21arxiv} as a potential candidate for achieving uncloneable-indistinguishable security without computational assumptions. While a proof of \emph{strong} uncloneable-indistinguishable security for this scheme remains an open question—\cite{BC25arXiv} established a \emph{weaker} variant with inverse-logarithmic success probability—in this work, we demonstrate that this scheme satisfies both untelegraphable-indistinguishable security and $t$-copy untelegraphable-indistinguishable security.

\begin{definition}
    Let $r,n\in\N$. The \emph{rank-$r$ Haar-measure encryption of $n$ messages} is the QECM $\ttt{Q}_{n,r}=([n],\mc{U}(rn),\mu_{Haar},\C^{rn},\{U\sigma_iU^\ast\}_{U\in\mc{U}(rn),i\in[n]})$, where $\sigma_i=\frac{1}{r}\sum_{j=ri}^{r(i+1)-1}\ketbra{j}$. We call the rank-$r$ Haar-measure encryption of $2$ messages the \emph{rank-$r$ Haar-measure encryption of a bit}
\end{definition}

We can also express $\sigma_i=\frac{1}{r}\ketbra{i}\otimes I_r$, following from the isomorphism $\C^{rn}\cong\C^n\otimes\C^r$.

\begin{remark} \label{rem:efficient_construction}
    The Haar-random QECM scheme can be made computationally efficient by replacing the Haar-random unitary with a unitary sampled from a $t$-design. As will become evident from the proof below, achieving $t$-copy telegraphing security requires the use of a unitary $2t$-design. Such designs admit efficient implementations, as demonstrated in \cite{MPSY24}. An explicit construction based on a unitary $2$-design is presented in \cite{BC25arXiv}.
\end{remark}

The telegraphing-distinguishing value of an attack $\ttt{A}=(\{m_0,m_1\},B,\{P^{k}_{m}\}_{k\in K,m\in M_0},\Phi)$ against the Haar-random QECM scheme is given by \Cref{def:cloning_attack} as
\begin{align*}
    \mfk{c}^2_{1 \rightarrow 1} (\ttt{Q}|\ttt{A}) = \int_{\mc{U}(rn)} \frac{1}{2} \sum_{b=0}^1 \Tr \squ{P^U_{m_b} \cdot \Phi(U \sigma_{m_b} U^*)} dU.
\end{align*}
where the integral is taken over normalized Haar measure of the unitary group. Let $\rho$ denote the positive operator corresponding to the \CPTP map $\Phi$ defined via the Choi–Jamiołkowski isomorphism as: $\rho \coloneqq (\id \otimes \Phi) \sum_{ij} \ketbra{ii}{jj}$. In terms of $\rho$, the telegraphing-distinguishing value can be rewritten as
\begin{equation*}
     \mfk{c}^2_{1 \rightarrow 1} (\ttt{Q}|\ttt{A})
     = \int \frac{1}{2} \sum_{b=0}^1 \Tr \squ{\rho ((U \sigma_{m_b} U^*)^\T \otimes P^U_{m_b})} dU
     = \int \frac{1}{2} \sum_{b=0}^1 \Tr \squ{\rho ((\bar{U} \sigma_{m_b} U^\T) \otimes P^U_{m_b})} dU.
\end{equation*}

\subsection{One copy untelegraphable-indistinguishable security}

We first establish the case of $1$-copy untelegraphable-indistinguishable security for Haar-measure encryption of a single classical bit. The full generalization to arbitrary message lengths and multiple-copy will be addressed subsequently.

\begin{theorem} \label{thm:haar_random_encryption_scheme_untelegraphable_indistinguishable}
    The rank-$d/2$ Haar-measure encryption of classical bits (\ie $2$ messages) achieves strong untelegraphable-indistinguishable security, with telegraphing-distinguishing value upper bounded as
    \begin{equation*}
        \mfk{c}^2_{1 \rightarrow 1} (\ttt{Q}|\scr{M})\leq\tfrac{1}{2} + \tfrac{1}{2\sqrt{d + 1}}.
    \end{equation*}
\end{theorem}
\begin{proof}
    Let $\Pi_i$ be the projection $\Pi_i \coloneqq \ketbra{i}{i} \otimes I_{d/2}$, such that the ciphertexts become $\sigma_i = \frac{2}{d} \cdot \Pi_i$ 
    The telegraphing-distinguishing value is given by:
    \begin{align*}
        \mfk{c}^2_{1 \rightarrow 1} (\ttt{Q}|\scr{M})
        &= \sup_{\ttt{A}} \int \frac{1}{2} \sum_{b=0}^1 \Tr \squ{\rho ((\bar{U} \sigma_{m_b} U^\T) \otimes P^U_{m_b})} dU \\
        &= \sup_{\ttt{A}} \int \sum_{b=0}^1 \Tr \squ{\tfrac{1}{d}\rho \cdot ((\bar{U} \Pi_{m_b} U^\T) \otimes P^U_{m_b})} dU,
    \end{align*}
    where the supremum is taken over all telegraphing-distinguishing attacks. Since $\frac{1}{d} \, \rho$ is now a normalised density operator, the optimisation can be relaxed by substituting $\tilde{\rho} = \frac{1}{d} \rho$ and optimising over all quantum states:
    \begin{equation*}
        \mfk{c}^2_{1 \rightarrow 1} (\ttt{Q}|\scr{M}) \leq \sup_{\tilde{\ttt{A}}} \int \sum_{b=0}^1 \Tr \squ{\tilde{\rho} \cdot ((\bar{U} \Pi_{m_b} U^\T) \otimes P^U_{m_b})}
    \end{equation*}
    Since the CPTP map $\Phi$ corresponding to the telegraphing attack is entanglement-breaking, its Choi matrix $\rho$ as well as the associated state $\tilde{\rho}$ are necessarily separable. As the optimization is over the convex set of separable states $\tilde{\rho}$, the supremum is attained at the extremal points of this set, namely the pure product states $\tilde{\rho} = \tilde{\rho}_1 \otimes \tilde{\rho}_2$. Consequently, we may write
    \begin{align*}
        \mfk{c}^2_{1 \rightarrow 1} (\ttt{Q}|\scr{M}) 
        &\leq \sup_{\tilde{\ttt{A}}} \int \sum_{b=0}^1 \Tr \squ{\tilde{\rho}_1 \otimes \tilde{\rho}_2 \cdot ((\bar{U} \Pi_{m_b} U^\T) \otimes P^U_{m_b})} \\
        &= \sup_{\tilde{\ttt{A}}} \int \sum_{b=0}^1  \underbrace{\Tr\squ{\tilde{\rho}_2 \cdot P^U_{m_b}}}_{\coloneqq q(U,m_b)} \Tr \squ{\tilde{\rho}_1 \cdot (\bar{U} \Pi_{m_b} U^\T)} dU,
    \end{align*}
    Since the $q(U,m_b)$ are non-negative reals and sum to $1$ over $b$, then by convexity
    \begin{equation*}
        \sum_{b=0}^1 q(U,m_b) \cdot \Tr \squ{\tilde{\rho}_1 \cdot (\bar{U} \Pi_{m_b} U^\T)}
        \leq \max_b \Big\{ \underbrace{\Tr \squ{\tilde{\rho}_1 \cdot (\bar{U} \Pi_{m_b} U^\T)}}_{\coloneqq M_b} \Big\}.
    \end{equation*}
    Applying the identity $\max \{ a, b \} = \frac{a+b}{2} + \frac{|a-b|}{2}$ yields:
    \begin{equation*}
        \mfk{c}^2_{1 \rightarrow 1} (\ttt{Q}|\scr{M}) \leq
        \sup_{\tilde{\ttt{A}}} \tfrac{1}{2} \int M_0 + M_1  dU + \tfrac{1}{2} \int | M_0 - M_1 | dU.
    \end{equation*}
    Now observe that $\Pi_0 + \Pi_1 = I_d$; thus we have
    \begin{equation*}
        \int M_0 + M_1 dU = \int \Tr \squ{\tilde{\rho}_1 \cdot (\bar{U} (\Pi_0 + \Pi_1) U^\T)} dU
        = 1,
    \end{equation*}
    so that:
    \begin{equation*}
        \mfk{c}^2_{1 \rightarrow 1} (\ttt{Q}|\scr{M})
        \leq \sup_{\tilde{\ttt{A}}} \tfrac{1}{2} + \tfrac{1}{2} \int | M_0 - M_1 | dU
        \leq \tfrac{1}{2} + \tfrac{1}{2} \sqrt{ \int (M_0 - M_1)^2 dU },
    \end{equation*}
    where the final inequality follows from Jensen's inequality. To evaluate the second moment, we write:
    \begin{equation*}
        {\left( M_0 - M_1 \right)}^2 = \Tr \Big[ \big( \tilde{\rho}_1 \otimes \tilde{\rho}_1 \big) \cdot \big( \bar{U} \otimes \bar{U} \big) {\big( \Pi_0 - \Pi_1 \big)}^{\otimes 2} \big( U^\T \otimes U^\T \big) \Big].
    \end{equation*}
    Applying the Weingarten calculus for the second moment (see e.g. \cite[Cor. 13]{Mel24}), we obtain:
    \begin{equation*}
        \int \big( \big( \bar{U} \otimes \bar{U} \big) {\big( \Pi_0 - \Pi_1 \big)}^{\otimes 2} \big( U^\T \otimes U^\T \big) dU = c_{\operatorname{I}} \cdot \operatorname{I} \,+\, c_{\operatorname{F}} \cdot \operatorname{F},
    \end{equation*}
    where $\operatorname{I}$ and $\operatorname{F}$ denote respectively the identity and flip (swap) operator on $\CC^d \otimes \CC^d$, defined by $\operatorname{I}(\ket{a} \otimes \ket{b}) = \ket{a} \otimes \ket{b}$ and $\operatorname{F}(\ket{a} \otimes \ket{b}) = \ket{b} \otimes \ket{a}$. The corresponding Weingarten coefficients $c_{\operatorname{I}}$ and $c_{\operatorname{F}}$ are given by:
    \begin{align*}
        c_{\operatorname{I}} &= \frac{\Tr[(\Pi_0 - \Pi_1)^{\otimes 2} \cdot \operatorname{I}] - \frac{1}{d}\Tr[(\Pi_0 - \Pi_1)^{\otimes 2} \cdot \operatorname{F}]}{d^2-1} = \frac{-1}{d^2 - 1}, \\
        c_{\operatorname{F}} &= \frac{\Tr[(\Pi_0 - \Pi_1)^{\otimes 2} \cdot \operatorname{F}] - \frac{1}{d}\Tr[(\Pi_0 - \Pi_1)^{\otimes 2} \cdot \operatorname{I}]}{d^2-1} = \frac{d}{d^2 - 1}.
    \end{align*}
    The operator $\Pi_0 - \Pi_1$ is a traceless Hermitian unitary operator on $\CC^d$. Using the swap trick $\Tr[(A \otimes B) \cdot \operatorname{F}] = \Tr[A \cdot B]$, we conclude:
    \begin{equation*}
        \int {\left( M_0 - M_1 \right)}^2 dU = \frac{-1}{d^2 - 1} \Tr {\left[ \tilde{\rho}_1 \right]}^2 + \frac{d}{d^2 - 1}\Tr \left[ \tilde{\rho}_1^2 \right] \leq \frac{1}{d+1},
    \end{equation*}
    where we have used the facts that $\Tr[\tilde{\rho}_1]^2 = 1$ and $\Tr[\tilde{\rho}_1^2] \leq 1$ for all quantum states $\tilde{\rho}_1$. Hence, we obtain the claimed upper bound: $\tfrac{1}{2} + \tfrac{1}{2\sqrt{d + 1}}$.
\end{proof}

\subsection{Many copy untelegraphable-indistinguishable security}

The above proof relies on the exact expression of the second unitary moment operator, derived via Weingarten calculus. In contrast, the proof of $t$-copy untelegraphable-indistinguishable security requires the analysis of higher-order moments. While exact evaluations are feasible for low-order moments, they become intractable as the order increases. To address this, we approximate the higher-order moments using the following lemma.

\begin{lemma}[{\cite[Lem. 1]{SHH24arxiv}}] \label{lem:weingarten_calculus_bounds}
    Let $\Phi_k$ and $\Psi_k$ be two hermitian-preserving maps from $\mc{B}(\CC^{d^k})$ to $\mc{B}(\CC^{d^k})$ defined by
    \begin{align*}
        \Phi_k (X) &\coloneqq \int_{\mathcal{U}(d)} U^{\otimes k} \, X \, {U^*}^{\otimes k} \operatorname{d}U \stackrel{\substack{\text{\tiny weingarten} \\ \text{\tiny calculus} \\[1ex]}}{=} \sum_{\pi,\sigma \in \mathfrak{S}_k} \operatorname{Wg}(\pi^{\shortminus 1} \sigma, d) \Tr \big[ V_d(\sigma)^{-1} X \big] \cdot V_d(\pi) \\[1em]
        \Psi_k (X) &\coloneqq \tfrac{1}{d^k} \sum_{\pi \in \mathfrak{S}_k} \Tr \big[ V_d(\pi)^{\shortminus 1} X \big] \cdot V_d(\pi),
    \end{align*}
    where $\mathfrak{S}_k$ denotes the symmetric group on $k$ elements, $\operatorname{Wg}(\cdot,\cdot)$ the Weingarten function, and $V_d(\pi)$ is defined as the tensor permutation of $(\CC^d)^{\otimes k}$ associated with $\pi \in \mathfrak{S}_k$. If $d > \sqrt{6}k^{7/4}$, then we have the inequalities
    \begin{equation*}
        \Big( 1 - \tfrac{k^2}{d} \Big) \Psi_k \preceq \Phi_k \preceq \Big( 1 + \tfrac{k^2}{d} \Big) \Psi_k,
    \end{equation*}
    where $\preceq$ is the order on hermitian-preserving maps given by $\Phi\preceq \Psi$ if $(\id\otimes\Phi)(P)\leq(\id\otimes\Psi)(P)$ for all positive semidefinite $P$.
\end{lemma}

The value of a $t$-copy telegraphing-distinguishing attack against the Haar-random QECM scheme $\ttt{A}=(\{m_0,m_1\},X,\{P_x\}_{x\in X},\{p(0|x,k),p(1|x,k)\}_{x\in X,k\in K})$ is given by \cref{def:cloning_attack,lem:tg-char} as:
\begin{align*}
    \mfk{c}^2_{t \rightarrow 1} (\ttt{Q}|\ttt{A}) = \int_{\mc{U}(rn)} \frac{1}{2} \sum_{b\in\{0,1\}} \sum_{x\in X} p(b|x,U) \Tr\squ{P_x \cdot (U \sigma_{m_b} U^*)^{\otimes t}} dU.
\end{align*}
where the integral is taken over normalized Haar measure of the unitary group. We can now prove the $t$-copy untelegraphable-indistinguishable security for Haar-measure encryption of any number of messages.

\begin{theorem} \label{thm:haar_random_encryption_scheme_t_copy_untelegraphable_indistinguishable}
    The rank-$r$ Haar-measure encryption of $n$ messages achieves strong untelegraphable-indistinguishable security, with $t$-copy telegraphing-distinguishing value upper bounded as
    \begin{equation*}
        \mfk{c}^2_{t \rightarrow 1} (\ttt{Q}_{r,n}|\scr{M})\leq\tfrac{1}{2} + \tfrac{7t}{\sqrt{r}}.
    \end{equation*}
\end{theorem}

\begin{lemma} \label{lem:haar_random_encryption_scheme_t_copy_untelegraphable_indistinguishable_lemma1}
    Let $P \in \mc{B}(\CC^{d^k})$ such that $0 \leq P \leq I$, and let $\sigma\in\mc{B}(\C^d)$ be a state such that $\norm{\sigma}\leq\varepsilon$. Suppose $k^2\leq \frac{1}{\varepsilon}$, then
    \begin{align*}
        \left| \int \Tr \left[ P\cdot {(U \sigma U^*)}^{\otimes k} \right] dU -\frac{\Tr[P]}{d^k} \right| \leq \frac{\Tr[P]}{d^k} 7k^2\varepsilon.
    \end{align*}
\end{lemma}

It follows from the lemma that $\norm*{\int_{\mc{U}(d)}(U\sigma U^\ast)^{\otimes k}dU-\frac{I}{d^k}}_{\Tr}\leq 7k^2\varepsilon$, and that if $\sigma \coloneqq \frac{1}{r} \Pi$ for $\Pi\in\mc{B}(\C^d)$ a rank $r$ projector and $k^2\leq r$, then
    \begin{align*}
        \left| \int \Tr \left[ P\cdot {(U \sigma U^*)}^{\otimes k} \right] dU -\frac{\Tr[P]}{d^k} \right| \leq \frac{\Tr[P]}{d^k} \frac{7k^2}{r}.
    \end{align*}

\begin{proof}
    First, using \cref{lem:weingarten_calculus_bounds}, since $k^2\leq\frac{1}{\varepsilon}\leq d$, we see that
    \begin{align*}
        \int \Tr \left[ P\cdot {(U\sigma U^*)}^{\otimes k} \right] dU
        &\leq \left( 1 + \frac{k^2}{d} \right) \frac{1}{d^k} \sum_{\pi \in \mathfrak{S}_k} \Tr \left[ P \cdot V_d(\pi) \right]\prod_{i=1}^{\#\pi}\Tr(\sigma^{n_i(\pi)}),
    \end{align*}
    where $\#\pi$ is the number of cycles in the cycle decomposition of $\pi$ and $(n_1(\pi),\ldots,n_{\#\pi}(\pi))$ is the cycle shape. We have that $\Tr(\sigma^n)\leq\Tr(\norm{\sigma}^{n-1}\sigma)=\norm{\sigma}^{n-1}$. Using this and the upper bound $\operatorname{Re} \Tr[P \cdot V_d(\pi)] \leq \Tr[P]$,
    \begin{align*}
        \int \Tr \left[ P\cdot {(U\sigma U^*)}^{\otimes k} \right] dU
        &\leq\parens*{1+\frac{k^2}{d}}\frac{\Tr[P]}{d^k}\sum_{\pi\in \mfk{S}_k}\prod_{i=1}^{\#\pi}\Tr(\sigma^{n_i(\pi)})\\
        &\leq\parens*{1+\frac{k^2}{d}}\frac{\Tr[P]}{d^k}\sum_{\pi\in \mfk{S}_k}\varepsilon^{\sum_{i=1}^{\#\pi}(n_i(\sigma)-1)}\\
        &\leq\parens*{1+\frac{k^2}{d}}\Tr[P]\parens*{\frac{\varepsilon}{d}}^k\sum_{\pi\in \mfk{S}_k}\varepsilon^{-\#\pi}.
    \end{align*}
    Then we use that
    \begin{equation*}
        \sum_{\pi \in \mathfrak{S}_k} x^{\#\pi} = \sum^k_{i=0} \genfrac{[}{]}{0pt}{1}{k}{i} x^i = x^{\bar{k}},
    \end{equation*}
    where $\genfrac{[}{]}{0pt}{1}{k}{i}$ denotes the unsigned Stirling numbers of the first kind, and $x^{\bar{k}}$ is the rising factorial
    \begin{equation*}
        x^{\bar{k}} \coloneqq x (x + 1) \cdots (x + k - 1).
    \end{equation*}
    Then
    \begin{align*}
        \int \Tr \left[ P\cdot {(U\sigma U^*)}^{\otimes k} \right] dU
        &\leq \left( 1 + \frac{k^2}{d} \right)\Tr[P]\parens*{\frac{\varepsilon}{d}}^k \varepsilon^{-1} (\varepsilon^{-1} + 1) \cdots (\varepsilon^{-1} + k - 1) \\
        &\leq \left( 1 + \frac{k^2}{d} \right) \frac{\Tr[P]}{d^k} \big( 1 + \varepsilon \big) \cdots \big(  1 + (k-1)\varepsilon \big) \\
        &\leq \left( 1 + \frac{k^2}{d} \right) \frac{\Tr[P]}{d^k} \prod^{k-1}_{i=0} \Big( 1 + i\varepsilon \Big).
    \end{align*}
    Then using the assumption $k^2 \leq \varepsilon^{-1}$
    \begin{equation*}
        \prod^{k-1}_{i=0} \Big( 1 + i\varepsilon \Big) \leq {\Big( 1 + k\varepsilon \Big)}^k \leq 1 + {\Big( 1 + \frac{1}{k} \Big)}^k k^2\varepsilon \leq 1 + e k^2\varepsilon.
    \end{equation*}
    Thus we have the upper bound
    \begin{align*}
        \int \Tr \left[ P\cdot {(U\sigma U^*)}^{\otimes k} \right] dU
        &\leq \left( 1 + \frac{k^2}{d} \right) \frac{\Tr[P]}{d^k} \left( 1 + e k^2\varepsilon \right) \\
        &\leq \frac{\Tr[P]}{d^k} \left( 1 + (2e + 1) k^2\varepsilon \right).
    \end{align*}
    On the other hand, using \cref{lem:weingarten_calculus_bounds}, we can lower bound
    \begin{align*}
        \int \Tr \left[ P\cdot {(U\sigma U^*)}^{\otimes k} \right] dU
        &\geq \left( 1 - \frac{k^2}{d} \right) \frac{1}{d^k} \sum_{\pi \in \mathfrak{S}_k} \Tr \left[ V_d(\pi)^{\shortminus 1} \cdot \sigma^{\otimes k} \right] \Tr \left[ P \cdot V(\pi) \right] \\
        &\geq \parens*{1-\frac{k^2}{d}}\frac{\Tr(P)}{d^k}\parens*{2-\sum_{\pi\in S_k}\Tr\squ*{V_d(\pi)^\ast\sigma^{\otimes k}}}\\
        &\geq\parens*{1-\frac{k^2}{d}}\frac{\Tr(P)}{d^k}\parens*{2-\parens*{1+k\varepsilon}^k}\\
        &\geq\frac{\Tr(P)}{d^k}\parens*{1-\frac{k^2}{d}}\parens*{1-ek^2\varepsilon}\\
        &\geq\frac{\Tr(P)}{d^k}\parens*{1-(e+1)k^2\varepsilon}.\qedhere
    \end{align*}
\end{proof}
\begin{proof}[Proof of \cref{thm:haar_random_encryption_scheme_t_copy_untelegraphable_indistinguishable}]
    Let $\ttt{A}=(\{m_0,m_1\},B,\{P^k_m\},\Phi)$ be a $t$-copy telegraphing attack against $\ttt{Q}_{r,n}$. Using \cref{lem:tg-char}, there exists a finite set $X$, probability distributions $p(\cdot|x,U)$ on $[2]$ for all $x\in X$ and $U\in\mc{U}(rn)$, and a POVM $\{P_x\}_{x\in X}$ such that
    \begin{equation*}
        \mfk{c}^2_{t \rightarrow 1} (\ttt{Q}|\ttt{A})
        = \int \frac{1}{2r} \sum_{b\in\{0,1\}} \sum_{x\in X} p(b|x,U) \Tr\squ{P_x \cdot (U \Pi_{m_b} U^*)^{\otimes t}} dU,
    \end{equation*}
    where $\Pi_i \coloneqq \ketbra{i}{i} \otimes I_{r}$. We want to upper-bound this value. To achieve this, it suffices to establish an upper bound on the following quantity:
    \begin{equation*}
        D \coloneqq \left| \tfrac{1}{r} \int \sum_x p(0|x,U) \cdot \Tr \left[ P_x \cdot {\left( U \, \Pi_{m_0} \, U^* \right)}^{\otimes t} \right] dU - \tfrac{1}{r} \int \sum_x p(0|x,U) \cdot \Tr \left[ P_x \cdot {\left( U \, \Pi_{m_1} \, U^* \right)}^{\otimes t} \right] dU \right|.
    \end{equation*}
    First using the triangle inequality,
    \begin{align*}
        D
        &= \bigg| \tfrac{1}{r} \int \sum_x p(0|x,U) \cdot \left( \Tr \left[ P_x \cdot {\left( U \, \Pi_{m_0} \, U^* \right)}^{\otimes t} \right] - \frac{\Tr[P_x]}{d^t} \right) dU \\ &\quad - \tfrac{1}{r} \int \sum_x p(0|x,U) \cdot \left( \Tr \left[ P_x \cdot {\left( U \, \Pi_{m_1} \, U^* \right)}^{\otimes t} \right] - \frac{\Tr[P_x]}{d^t} \right) dU \bigg| \\
        &\leq \bigg| \tfrac{1}{r} \int \sum_x p(0|x,U) \cdot \left( \Tr \left[ P_x \cdot {\left( U \, \Pi_{m_0} \, U^* \right)}^{\otimes t} \right] - \frac{\Tr[P_x]}{d^t} \right) dU \bigg| \\ &\quad + \bigg| \tfrac{1}{r} \int \sum_x p(0|x,U) \cdot \left( \Tr \left[ P_x \cdot {\left( U \, \Pi_{m_1} \, U^* \right)}^{\otimes t} \right] - \frac{\Tr[P_x]}{d^t} \right) dU \bigg|,
    \end{align*}
    where we write $d\coloneqq rn$. Therefore, for any $m_b$, we can bound by Cauchy-Schwarz as
    \begin{align*}
        &\bigg| \tfrac{1}{r} \int \sum_x p(0|x,U) \cdot \left( \Tr \left[ P_x \cdot {\left( U \, \Pi_{m_b} \, U^* \right)}^{\otimes t} \right] - \frac{\Tr[P_x]}{d^t} \right) dU \bigg| \\
        &\quad\leq \sum_x  \bigg| \tfrac{1}{r} \int p(0|x,U) \cdot \left( \Tr \left[ P_x \cdot {\left( U \, \Pi_{m_b} \, U^* \right)}^{\otimes t} \right] - \frac{\Tr[P_x]}{d^t} \right) dU \bigg| \\
        &\quad\leq \sum_x \sqrt{ \int p(0|x,U)^2 dU \cdot \int \left( \Tr \left[ P_x \cdot {\left( U \left(\tfrac{1}{r} \cdot \Pi_{m_b} \right) U^* \right)}^{\otimes t} \right] - \frac{\Tr[P_x]}{d^t} \right)^2 dU } \\
        &\quad\leq \sum_x \sqrt{ \int \left( \Tr \left[ P_x \cdot {\left( U \left(\tfrac{1}{r} \cdot \Pi_{m_b} \right) U^* \right)}^{\otimes t} \right] - \frac{\Tr[P_x]}{d^t} \right)^2 dU }.
    \end{align*}
    The inner term expands to
    \begin{align*}
        &\left( \Tr \left[ P_x \cdot {\left( U \left(\tfrac{1}{r} \cdot \Pi_{m_b} \right) U^* \right)}^{\otimes t} \right] - \frac{\Tr[P_x]}{d^t} \right)^2 \\
        &\quad = \Tr \left[ \left( P_x \otimes P_x \right) \cdot {\left( U \left(\tfrac{1}{r} \cdot \Pi_{m_b} \right) U^* \right)}^{\otimes 2t} \right] - 2 \frac{\Tr[P_x]}{d^t} \Tr \left[ P_x \cdot {\left( U \left(\tfrac{1}{r} \cdot \Pi_{m_b} \right) U^* \right)}^{\otimes t} \right] + \frac{\Tr[P_x]^2}{d^{2t}}.
    \end{align*}
    Then, using \cref{lem:haar_random_encryption_scheme_t_copy_untelegraphable_indistinguishable_lemma1} for sufficiently large $d$,
    \begin{equation*}
        \int \Tr \left[ \left( P_x \otimes P_x \right) \cdot {\left( U \left(\tfrac{1}{r} \cdot \Pi_{m_b} \right) U^* \right)}^{\otimes 2t} \right] dU \leq \frac{\Tr[P_x \otimes P_x]}{d^{2t}} \left( 1 + \frac{7(2t)^2}{r} \right) = \frac{\Tr[P_x]^2}{d^{2t}} \left( 1 + 28 \frac{t^2}{r} \right),
    \end{equation*}
    and
    \begin{equation*}
        \Tr \left[ P_x \cdot {\left( U \left(\tfrac{1}{r} \cdot \Pi_{m_b} \right) U^* \right)}^{\otimes t} \right] \geq \frac{\Tr[P_x]}{d^t} \left( 1 - 7 \frac{t^2}{r} \right).
    \end{equation*}
    Putting these together,
    \begin{align*}
        &\int \left( \Tr \left[ P_x \cdot {\left( U \left(\tfrac{1}{r} \cdot \Pi_{m_b} \right) U^* \right)}^{\otimes t} \right] - \frac{\Tr[P_x]}{d^t} \right)^2 dU \\
        &\quad\leq \frac{\Tr[P_x]^2}{d^{2t}} \left( 1 + 28 \frac{t^2}{r} \right) - 2\frac{\Tr[P_x]^2}{d^{2t}} \left( 1 - 7 \frac{t^2}{r} \right) + \frac{\Tr[P_x]^2}{d^{2t}} \\
        &\quad\leq 42 \frac{ t^2}{r} \frac{\Tr[P_x]^2}{d^{2t}},
    \end{align*}
    and get
    \begin{equation*}
        \bigg| \tfrac{1}{r} \int \sum_x p(0|x,U) \cdot \left( \Tr \left[ P_x \cdot {\left( U \, \Pi_{m_b} \, U^* \right)}^{\otimes t} \right] - \frac{\Tr[P_x]}{d^t} \right) dU \bigg| \leq \frac{7t}{\sqrt{r}}.
    \end{equation*}
    As this holds for both $b$, then $D \leq \frac{14t}{\sqrt{r}}$.

    Now we can write the $t$-copy telegraphing-distinguishing value $\mfk{c}^2_{t \rightarrow 1} (\ttt{Q}|\ttt{A})$ as
    \begin{equation*}
        \mfk{c}^2_{t \rightarrow 1} (\ttt{Q}|\ttt{A}) = \tfrac{1}{2} \Pr \left[ b' = 0 | b = 0 \right] + \tfrac{1}{2} \Pr \left[ b' = 1 | b = 1 \right],
    \end{equation*}
    where $b'$ is the random variable corresponding to the adversary's guess of $b$. Then
    \begin{align*}
        \mfk{c}^2_{t \rightarrow 1} (\ttt{Q}|\ttt{A})
        &\leq \tfrac{1}{2} \Pr \left[ b' = 0 | b = 1 \right] + \tfrac{1}{2} \Pr \left[ b' = 1 | b = 1 \right] + \tfrac{1}{2} D \\
        &\leq \Pr \left[ \left( b' = 0 \wedge b = 1 \right) \vee \left( b' = 1 \wedge b = 1 \right) \right] + \tfrac{1}{2} D \\
        &\leq \tfrac{1}{2} + \tfrac{1}{2} D.\qedhere
    \end{align*}
\end{proof}
\section{Collusion-resistant security} \label{sec:collusion}

In this section, we adapt the techniques of \cref{sec:haar_random} to the setting of collusion-resistant security. In doing so, we find that the same bounds hold on telegraphing-distinguishing attacks with collusion. We use this to strengthen some results of~\cite{CKNY24arxiv} by removing or weakening computational assumptions.

\begin{lemma}\label{lem:different-states}
    Let $P \in \mc{B}(\CC^{d^k})$ such that $0 \leq P \leq I$, and let $\sigma_1,\ldots,\sigma_n\in\mc{B}(\C^d)$ be orthogonal states such that $\norm{\sigma_i}\leq\varepsilon$. Let $k=\sum_{i=1}^nk_i$ and suppose $k^2\leq \frac{1}{\varepsilon}$. Then
    \begin{align*}
        \left| \int \Tr \left[ P\cdot {(U \sigma_1 U^*)}^{\otimes k_1}\otimes\cdots\otimes {(U\sigma_nU^\ast)}^{\otimes k_n} \right] dU -\frac{\Tr[P]}{d^k} \right| \leq \frac{\Tr[P]}{d^k}7k^2\varepsilon.
    \end{align*}
\end{lemma}

\begin{proof}
    First, using \cref{lem:weingarten_calculus_bounds}, since $k^2\leq\frac{1}{\varepsilon}\leq d$, we see that
    \begin{align*}
        \int \Tr &\left[ P\cdot {(U \sigma_1 U^*)}^{\otimes k_1}\otimes\cdots\otimes {(U\sigma_nU^\ast)}^{\otimes k_n} \right] dU\\
        &\leq \left( 1 + \frac{k^2}{d} \right) \frac{1}{d^k} \sum_{\pi \in \mathfrak{S}_k} \Tr \left[ P \cdot V_d(\pi) \right]\Tr\squ*{V_d(\pi)^{-1}(\sigma_1^{\otimes k_1}\otimes\cdots\sigma_n^{\otimes k_n})}.
    \end{align*}
    Now, write the cycle decomposition of $\pi$ as $\pi=(i_{1,1}\;\ldots\;i_{1,n_1})\cdots(i_{\#\pi,1}\;\ldots\;i_{\#\pi,n_{\#\pi}})$ and let $f:\{1,\ldots,k\}\rightarrow\{1,\ldots,n\}$ be the function such that $\sigma_1^{\otimes k_1}\otimes\cdots\sigma_n^{\otimes k_n}=\sigma_{f(1)}\otimes\cdots\otimes\sigma_{f(k)}$. Then,
    \begin{align*}
        \Tr\squ*{V_d(\pi)^{-1}(\sigma_1^{\otimes k_1}\otimes\cdots\sigma_n^{\otimes k_n})}&=\prod_{j=1}^{\#\pi}\Tr\squ*{\sigma_{f(i_{j,1})}\cdots \sigma_{f(i_{j,n_j})}}\\
        &=\begin{cases}0&\exists\,j,l,m.\;f(i_{j,l})\neq f(i_{j,m})\\\prod_{j=1}^{\#\pi}\Tr\squ*{\sigma_{f(i_{j,1})}^{n_j}}&\text{ else}
        \end{cases}\\
        &\leq\varepsilon^{k-\#\pi}.
    \end{align*}
    Note also that $\Tr\squ*{V_d(\pi)^{-1}(\sigma_1^{\otimes k_1}\otimes\cdots\sigma_n^{\otimes k_n})}\geq 0$. Therefore, using the upper bound $\operatorname{Re} \Tr[P \cdot V_d(\pi)] \leq \Tr[P]$, and hence
    \begin{align*}
        \int \Tr \left[ P\cdot {(U \sigma_1 U^*)}^{\otimes k_1}\otimes\cdots\otimes {(U\sigma_nU^\ast)}^{\otimes k_n} \right] dU &\leq \left( 1 + \frac{k^2}{d} \right)\frac{\Tr[P]}{d^k}\sum_{\pi \in \mathfrak{S}_k}\Tr\squ*{V_d(\pi)^{-1}(\sigma_1^{\otimes k_1}\otimes\cdots\sigma_n^{\otimes k_n})}\\
        &\leq \left( 1 + \frac{k^2}{d} \right)\Tr[P]\parens*{\frac{\varepsilon}{d}}^k\sum_{\pi \in \mathfrak{S}_k}\varepsilon^{-\#\pi}.
    \end{align*}
    So, proceeding exactly as in \cref{lem:haar_random_encryption_scheme_t_copy_untelegraphable_indistinguishable_lemma1}, we get the upper bound
    \begin{align*}
         \int \Tr \left[ P\cdot {(U \sigma_1 U^*)}^{\otimes k_1}\otimes\cdots\otimes {(U\sigma_nU^\ast)}^{\otimes k_n} \right] dU&\leq \frac{\Tr[P]}{d^k} \left( 1 + (2e + 1) k^2\varepsilon \right).
    \end{align*}
    For the lower bound, we proceed the same way, getting
    \begin{align*}
        \int \Tr &\left[ P\cdot {(U \sigma_1 U^*)}^{\otimes k_1}\otimes\cdots\otimes {(U\sigma_nU^\ast)}^{\otimes k_n} \right] dU\\
        &\geq \left( 1 - \frac{k^2}{d} \right) \frac{1}{d^k} \sum_{\pi \in \mathfrak{S}_k} \Tr \left[ V_d(\pi)^{\shortminus 1} \cdot (\sigma_1^{\otimes k_1}\otimes\cdots\sigma_n^{\otimes k_n}) \right] \Tr \left[ P \cdot V_d(\pi) \right] \\
        &\geq \parens*{1-\frac{k^2}{d}}\frac{\Tr(P)}{d^k}\parens*{2-\sum_{\pi\in S_k}\Tr\squ*{V_d(\pi)^{-1}(\sigma_1^{\otimes k_1}\otimes\cdots\sigma_n^{\otimes k_n})}}\\
        &\geq\frac{\Tr(P)}{d^k}\parens*{1-(e+1)k^2\varepsilon}.\qedhere
    \end{align*}
\end{proof}

\begin{theorem}\label{thm:collusion}
    Let $\ttt{A}=(Q,X,\{p(b|x,U)\}_{b\in\{0,1\},k\in\mc{U}(rn),x\in X},\{H_i\}_{i\in[Q+1]},\{V^i_{m_0,m_1}\}_{i\in[Q];m_0,m_1\in M,m_0\neq m_1},\{P_x\}_{x\in X})$ be a telegraphing-distinguishing attack with collusion against $\ttt{Q}_{r,n}$. Then,
    $$\mfk{t}(\ttt{Q}_{r,n}|\ttt{A})\leq\frac{1}{2}+7\frac{Q}{\sqrt{r}}.$$
\end{theorem}

\begin{proof}
    Write $M^{(2)}=\set*{(m_0,m_1)\in M^2}{m_0\neq m_1}$ and let $P_{x,m^{(1)},\ldots,m^{(Q)}}=((V^{0}_{m^{(1)}})^\dag\otimes I)\cdots((V^{Q-1}_{m^{(Q)}})^\dag\otimes I)P_x(V^{Q-1}_{m^{(Q)}}\otimes I)\cdots(V^{0}_{m^{(1)}}\otimes I)$ for $x\in X$ and $m^{(1)},\ldots,m^{(Q)}\in M^{(2)}$. $\{P_{x,m^{(1)},\ldots,m^{(Q)}}\}_{x\in X;m^{(1)},\ldots,m^{(Q)}\in M^{(2)}}$ is a POVM on $(\C^{rn})^{\otimes Q}$, and the value of $\ttt{A}$ against $\ttt{Q}_{r,n}$ is
    \begin{align*}
        \mfk{t}(\ttt{Q}_{r,n}|\ttt{A})=\int_{\mc{U}(rn)}\frac{1}{2}\sum_{\substack{b\in\{0,1\},x\in X\\m^{(1)},\ldots,m^{(Q)}\in M^{(2)}}}p(b|x,U)\Tr\squ{P_{x,m^{(1)},\ldots,m^{(Q)}}(U\sigma_{m_b^{(1)}}U^\ast\otimes\cdots\otimes U\sigma_{m_b^{(Q)}}U^\ast)}dU.
    \end{align*}
    Now, we bound the quantity 
    \begin{align*}
        D\coloneqq\abs*{\int \ \  \sum_{\mathclap{\substack{x,\\m^{(1)},\ldots,m^{(Q)}}}} \ \  p(0|x,U)\Tr\squ{P_{x,m^{(1)},\ldots,m^{(Q)}}(U\sigma_{m_0^{(1)}}U^\ast\otimes\cdots\otimes U\sigma_{m_0^{(Q)}}U^\ast-U\sigma_{m_1^{(1)}}U^\ast\otimes\cdots\otimes U\sigma_{m_1^{(Q)}}U^\ast)}dU}.
    \end{align*}
    As in \cref{thm:haar_random_encryption_scheme_t_copy_untelegraphable_indistinguishable}, the value $\mfk{t}(\ttt{Q}_{r,n}|\ttt{A})\leq\frac{1}{2}+\frac{1}{2}D$. Then, using the triangle inequality, we can upper bound
    \begin{align*}
        D&\leq\abs*{\int \ \  \sum_{\mathclap{\substack{x,\\m^{(1)},\ldots,m^{(Q)}}}} \ \  p(0|x,U)\parens*{\Tr\squ{P_{x,m^{(1)},\ldots,m^{(Q)}}(U\sigma_{m_0^{(1)}}U^\ast\otimes\cdots\otimes U\sigma_{m_0^{(Q)}}U^\ast)}-\frac{\Tr[P_{x,m^{(1)},\ldots,m^{(Q)}}]}{d^Q}}dU}\\
        &\qquad+\abs*{\int \ \  \sum_{\mathclap{\substack{x,\\m^{(1)},\ldots,m^{(Q)}}}} \ \  p(0|x,U)\parens*{\Tr\squ{P_{x,m^{(1)},\ldots,m^{(Q)}}(U\sigma_{m_1^{(1)}}U^\ast\otimes\cdots\otimes U\sigma_{m_1^{(Q)}}U^\ast)}-\frac{\Tr[P_{x,m^{(1)},\ldots,m^{(Q)}}]}{d^Q}}dU}.
    \end{align*}
    Applying Cauchy-Schwarz as in \cref{thm:haar_random_encryption_scheme_t_copy_untelegraphable_indistinguishable},
    \begin{align*}
        &\abs*{\int \ \  \sum_{\mathclap{\substack{x,\\m^{(1)},\ldots,m^{(Q)}}}} \ \  p(0|x,U)\parens*{\Tr\squ{P_{x,m^{(1)},\ldots,m^{(Q)}}(U\sigma_{m_b^{(1)}}U^\ast\otimes\cdots\otimes U\sigma_{m_b^{(Q)}}U^\ast)}-\frac{\Tr[P_{x,m^{(1)},\ldots,m^{(Q)}}]}{d^Q}}dU}\\
        &\quad\leq\hspace{-0.4cm}\sum_{\substack{x,\\m^{(1)},\ldots,m^{(Q)}}}\abs*{\int p(0|x,U)\parens*{\Tr\squ{P_{x,m^{(1)},\ldots,m^{(Q)}}(U\sigma_{m_b^{(1)}}U^\ast\otimes\cdots\otimes U\sigma_{m_b^{(Q)}}U^\ast)}-\frac{\Tr[P_{x,m^{(1)},\ldots,m^{(Q)}}]}{d^Q}}dU}\\
        &\quad\leq\hspace{-0.4cm}\sum_{\substack{x,\\m^{(1)},\ldots,m^{(Q)}}}\sqrt{\int p(0|x,U)^2dU\cdot\int\parens*{\Tr\squ{P_{x,m^{(1)},\ldots,m^{(Q)}}(U\sigma_{m_b^{(1)}}U^\ast\otimes\cdots\otimes U\sigma_{m_b^{(Q)}}U^\ast)}-\frac{\Tr[P_{x,m^{(1)},\ldots,m^{(Q)}}]}{d^Q}}^2dU}\\
        &\quad\leq\hspace{-0.4cm}\sum_{\substack{x,\\m^{(1)},\ldots,m^{(Q)}}}\sqrt{\int\parens*{\Tr\squ{P_{x,m^{(1)},\ldots,m^{(Q)}}(U\sigma_{m_b^{(1)}}U^\ast\otimes\cdots\otimes U\sigma_{m_b^{(Q)}}U^\ast)}-\frac{\Tr[P_{x,m^{(1)},\ldots,m^{(Q)}}]}{d^Q}}^2dU}.
    \end{align*}
    Fix $x,m^{(1)},\ldots,m^{(Q)}$ and write $P=P_{x,m^{(1)},\ldots,m^{(Q)}}$ and $\sigma=\sigma_{m_b^{(1)}}\otimes\cdots\otimes \sigma_{m_b^{(Q)}}$. Expanding,
    \begin{align*}
        \parens*{\Tr\squ*{PU^{\otimes Q}\sigma (U^\ast)^{\otimes Q}}-\frac{\Tr[P]}{d^Q}}^2&=\Tr\squ*{(P\otimes P)U^{\otimes 2Q}(\sigma\otimes\sigma) (U^\ast)^{\otimes 2Q}}-2\frac{\Tr[P]}{d^Q}\Tr\squ*{PU^{\otimes Q}\sigma (U^\ast)^{\otimes Q}}+\frac{\Tr[P]^2}{d^Q}.
    \end{align*}
    Taking the integral and using \cref{lem:different-states}, we get
    \begin{align}
    \begin{split}\label{eq:depends-on-haar}
        \int\parens*{\Tr\squ*{PU^{\otimes Q}\sigma (U^\ast)^{\otimes Q}}-\frac{\Tr[P]}{d^Q}}^2dU&\leq\frac{\Tr[P\otimes P]}{d^{2Q}}\parens*{1+\frac{7(2Q)^2}{r}}-2\frac{\Tr[P]}{d^Q}\frac{\Tr[P]}{d^Q}\parens*{1-\frac{7Q^2}{r}}+\frac{\Tr[P]^2}{d^Q}\\
        &\leq\frac{\Tr[P]^2}{d^{2Q}}42\frac{Q^2}{r}.
    \end{split}
    \end{align}
    Using this,
    \begin{align*}
        &\sum_{\substack{x,\\m^{(1)},\ldots,m^{(Q)}}}\sqrt{\int\parens*{\Tr\squ{P_{x,m^{(1)},\ldots,m^{(Q)}}(U\sigma_{m_b^{(1)}}U^\ast\otimes\cdots\otimes U\sigma_{m_b^{(Q)}}U^\ast)}-\frac{\Tr[P_{x,m^{(1)},\ldots,m^{(Q)}}]}{d^Q}}^2dU}\\
        &\qquad\leq\sum_{\substack{x,\\m^{(1)},\ldots,m^{(Q)}}}\frac{\Tr[P_{x,m^{(1)},\ldots,m^{(Q)}}]}{d^Q}\sqrt{42}\frac{Q}{\sqrt{r}}\leq 7\frac{Q}{\sqrt{r}}.
    \end{align*}
    Hence, $D\leq14\frac{Q}{\sqrt{r}}$, giving the result.
\end{proof}

\begin{corollary} \label{cor:collusion}
    \begin{enumerate}[(i)]
        \item For any polynomial $p$, there exists an efficient QECM that is collusion-resistant untelegraphable-indistinguishable secure against any attack $\{\ttt{A}_\lambda\}_\lambda$ with number of rounds bounded as $Q^\lambda\leq p(\lambda)$.

        \item If pseudorandom unitaries exist, there exists an efficient QECM that is everlasting collusion-resistant untelegraphable-indistinguishable secure.
    \end{enumerate}
\end{corollary}

This corollary strengthens some results of~\cite{CKNY24arxiv}. First, their construction of a collusion-resistant untelegraphable-indistinguishable secure relies on pseudorandom functions, which may be a stronger assumption that pseudorandom unitaries. Next, to construct the weaker notion of a QECM that is secure against attacks with the number of rounds bounded by a fixed polynomial,~\cite{CKNY24arxiv} require pseudorandom states whereas we do it unconditionally --- nevertheless their construction has succinct keys, where we require polynomial-length keys. Finally, we are able to achieve everlasting security under standard computational assumptions, whereas~\cite{CKNY24arxiv} require a quantum random oracle model.

\begin{proof}
    Note first that in \cref{thm:collusion}, only Equation~\eqref{eq:depends-on-haar} depends on the measure being the Haar measure. As such, if we replace the question distribution with a different distribution on the unitaries that has the same (or similar) behaviour in~\eqref{eq:depends-on-haar}, then we can get the same (or similar) upper bound on the value. We show the results for messages of a fixed bit length $\ell\in\N$, but they extend identically to messages of polynomially-varying length in $\lambda$.
    \begin{enumerate}[(i)]
        \item Let $\ttt{Q}_{q,\ell,d}$ be the QECM that is identical to $\ttt{Q}_{2^q,2^{\ell}}$ except that the key distribution is the uniform distribution on unitaries corresponding to circuits of depth $d$ on $n=\ell+q$ qubits. Due to~\cite{Haf22}, random circuits in depth $f(n,t)=O(nt^{5+o(1)})$ are $t$-designs. Let $\ttt{Q}_\lambda=\ttt{Q}_{\lambda,\ell,f(\lambda+\ell,2p(\lambda))}$. Then, $\{\ttt{Q}_\lambda\}_\lambda$ is an efficient QECM by construction. Since the distribution on keys is a $2p(\lambda)$-design, Equation~\eqref{eq:depends-on-haar} holds as long as $Q^\lambda\leq p(\lambda)$ and hence for any collection of attacks $\{\ttt{A}_\lambda\}_\lambda$ such that $Q^\lambda\leq p(\lambda)$, $\mfk{t}(\ttt{Q}_\lambda|\ttt{A}_\lambda)\leq\frac{1}{2}+7\frac{Q^\lambda}{\sqrt{2^\lambda}}=\frac{1}{2}+\negl(\lambda)$.

        \item Let $G_\lambda$ be a pseudorandom unitary in dimension $2^\lambda$. Let $\ttt{Q}_\lambda$ be the same as $\ttt{Q}_{2^\lambda,2^\ell}$ except that the key distribution if $G_\lambda(k)$ on uniform input. Then $\{\ttt{Q}_\lambda\}_\lambda$ is an efficient QECM. Let $\{\ttt{A}_\lambda\}_\lambda$ be a telegraphing-distinguishing attack with collusion against $\{\ttt{Q}_\lambda\}_\lambda$ such that the measurements $V^{\lambda,i}_{m_0,m_1}$ and $P^\lambda_x$ are implemented efficiently. In \cref{thm:collusion}, everything before Equation~\eqref{eq:depends-on-haar} does not depend on the key distribution. In particular, we can still find an upper bound independent of $p(b|x,U)$ since the key distribution need not look Haar-random to the second-stage adversary. Hence sampling  $p(b|x,U)$ need not be efficient, which step allows for everlasting security. Also, replacing the Haar measure with the pseudorandom unitary, the upper bound in \eqref{eq:depends-on-haar} will increase at most negligibly because the integrand depends polynomially on the unitary. Therefore, we get the upper bound $\mfk{t}(\ttt{Q}_\lambda|\ttt{A}_\lambda)\leq\frac{1}{2}+7\frac{Q^\lambda}{\sqrt{2^\lambda}}+\negl(\lambda)=\frac{1}{2}+\negl(\lambda)$
    \end{enumerate}
\end{proof}
\section{Untelegraphable encryption as a limit of 
uncloneable encryption} \label{sec:UTE_as_UE_limit}

In this section, we show that the cloning value of a QECM tends to the telegraphing value as the number of adversaries increases. To do so, we use information theoretic tools from work of Brand\~{a}o~\cite{Bra08}, where they were used to study $\tsf{QMA}(2)$ proof systems.

\begin{definition}
    The \emph{entanglement entropy} of a pure state $\psi_{AB}=\ketbra{\psi}_{AB}$ is $E(\psi)=S(\psi_A)=S(\psi_B)$.

    The \emph{entanglement of formation} of a state $\rho_{AB}$ is $E_F(\rho)=\min_{\{p_i,\ket{\psi_i}\}_i}\sum_ip_iE(\psi_i)$, where the minimisation is over ensembles of pure states such that $\rho=\sum_ip_i\ketbra{\psi_i}$.

    The \emph{Henderson-Vedral measure} of a state $\rho_{AB}$ is $C^\leftarrow(\rho)=\max_{\{p_i,\rho_i\}}S(\rho_A)-\sum_ip_iS(\rho_i)$, where the maximisation is over ensembles such that there exists a POVM $\{P_i\}_i$ on register $B$ such that $p_i=\Tr(P_i\rho_B)$ and $\rho_i=\frac{1}{p_i}\Tr_B((I\otimes P_i)\rho_{AB})$.

    The \emph{convex roof of the Henderson-Vedral measure} of a state $\rho_{AB}$ is $G^{\leftarrow}(\rho)=\min_{\{p_i,\rho_i\}}\sum_ip_iC^{\leftarrow}(\rho_i)$, where the minimisation is over ensembles of states such that $\rho=\sum_ip_i\rho_i$.
\end{definition}

\begin{lemma}[Yang's monogamy inequality~\cite{Yan06}]
Let $A,B_1,\ldots,B_N$ be registers and let $\rho_{AB_1\cdots B_N}$ be a quantum state. Then,    
$$E_{F}(\rho_{A;B_1\cdots B_N})\geq\sum_{i=1}^NG^{\leftarrow}(\rho_{AB_i}).$$
\end{lemma}

\begin{theorem} \label{thm:UTE_UE_asymptotic_equivalence}
    Let $\ttt{Q}=(M,K,\pi,H,\{\sigma^k_m\}_{k,m})$ be a QECM, and suppose $\eta\geq\int_K\norm{\sigma^k_m}d\pi(k)$ for all $m\in M$. Then, writing $d=\dim H$
    \begin{align*}
        \mfk{c}^N_{1\rightarrow s}(\ttt{Q})\leq\mfk{c}_{1\rightarrow 1}^N(\ttt{Q}|\scr{M})+3\eta d\parens*{\frac{\log d}{N^2s}}^{1/3}.
    \end{align*}
\end{theorem}

From here, it is direct to see that $\lim_{s\rightarrow\infty}\mfk{c}^N_{1\rightarrow s}(\ttt{Q})=\mfk{c}^N_{1\rightarrow 1}(\ttt{Q}|\scr{M})$. As a corollary, we also find that if $\norm{\sigma^k_m}\leq\eta$ for all $k,m$, then $\mfk{c}^N_{t\rightarrow s}(\ttt{Q})\leq\mfk{c}_{t\rightarrow 1}^N(\ttt{Q}|\scr{M})+3(\eta d)^t\parens*{\frac{t\log d}{N^2s}}^{1/3}$. For example, in the case of $\ttt{Q}=\ttt{Q}_{r,n}$ we can take $\eta=\frac{1}{r}$, so $$\mfk{c}^N_{t\rightarrow s}(\ttt{Q}_{r,n})\leq\mfk{c}_{t\rightarrow 1}^N(\ttt{Q}_{r,n}|\scr{M})+3n^t\parens*{\frac{t\log d}{N^2s}}^{1/3}\leq\frac{1}{N}+\frac{14 t}{N\sqrt{r}}+3n^t\parens*{\frac{t\log d}{N^2s}}^{1/3}.$$

\begin{proof}
Let $\ttt{A}=(M_0,\{B_i\}_i,\{P^{i,k}_m\}_{i,k,m},\Phi)$ be a $1$-to-$s$ $N$-message cloning attack against $\ttt{Q}$. First, we can write the cloning probability in an entanglement-based picture:
\begin{align*}
    \mfk{c}^N_{1\rightarrow s}(\ttt{Q}|\ttt{A})&=\int_K\frac{1}{N}\sum_{m\in M_0}\Tr\squ{(P^{1,k}_m\otimes\cdots\otimes P^{s,k}_m)\Phi(\sigma^k_m)}d\pi(k)\\
    &=\int_K\frac{1}{N}\sum_{m\in M_0}\sum_{i,j=1}^d\Tr\squ{(P^{1,k}_m\otimes\cdots\otimes P^{s,k}_m)\Phi(\ketbra{i}\sigma^k_m\ketbra{j})}d\pi(k)\\
    &=\int_K\frac{1}{N}\sum_{m\in M_0}\sum_{i,j=1}^d\Tr\squ{(\bar{\sigma}^k_m\otimes P^{1,k}_m\otimes\cdots\otimes P^{s,k}_m)(\ketbra{i}{j}\otimes\Phi(\ketbra{i}{j}))}d\pi(k)\\
    &=\int_K\frac{1}{N}\sum_{m\in M_0}\Tr\squ{(d\bar{\sigma}^k_m\otimes P^{1,k}_m\otimes\cdots\otimes P^{s,k}_m)J(\Phi)}d\pi(k),
\end{align*}
where $J(\Phi)$ is the Choi state of $\Phi$. By embedding into larger Hilbert spaces, we can assume $B_i=B$ for all $i$. Let $\rho_{AB^s}=\frac{1}{s!}\sum_{\tau\in \mfk{S}_s}(I\otimes V_B(\tau))J(\Phi)(I\otimes V_B(\tau))^\ast\otimes \ketbra{\tau(1)\cdots\tau(s)}$ and let $P^k_m=\sum_{i}P^{i,k}_m\otimes\ketbra{i}$. This achieves a labelled symmetrisation of the adversaries' registers, so

\begin{align*}
    \mfk{c}^N_{1\rightarrow s}(\ttt{Q}|\ttt{A})&=\int_K\frac{1}{N}\sum_{m\in M_0}\Tr\squ{(d\bar{\sigma}^k_m\otimes (P^k_m)^{\otimes s})\rho_{AB^s}}d\pi(k).
\end{align*}

Now, we know that $E_F(\rho_{A;B^s})\leq\log d$, where $d=\dim H$ and $\rho$ is symetric under permutations of the $B$ registers, so $\log d\geq sG^{\leftarrow}(\rho_{AB})$. Let $\{p_i,\rho_i\}_i$ be the optimising ensemble in the definition of the convex roof of the Henderson-Vedral measure, so $\sum_ip_i C^\leftarrow(\rho_i)\leq\frac{\log d}{s}$. Fix $\varepsilon>0$ and let $I=\set*{i}{C^\leftarrow(\rho_i)\geq\varepsilon}$. Then,
$$\frac{\log d}{s}\geq\sum_{i\in I}p_i\varepsilon.$$
Else, if $i\notin I$, $C^\leftarrow(\rho_i)\leq\varepsilon$. Consider the ensemble $\{p_{ikm},\rho_{ikm}\}_m$ induced by the measurement $P^k_m$. We must have $S(\rho_{i,A})-\sum_mp_{ikm}S(\rho_{ikm})\leq\varepsilon$. This is equal to the divergence $D(\sum_mp_{ikm}\rho_{ikm}\otimes\ketbra{m}||\rho_{i,A}\otimes\sum_mp_{ikm}\ketbra{m})$, so Pinsker's inequality implies $\sum_mp_{ikm}\norm{\rho_{i,A}-\rho_{ikm}}_{\Tr}\leq\sqrt{\frac{\varepsilon}{2}}$. Now, let $\tau_{AB}=\sum_ip_i\rho_{i,A}\otimes\rho_{i,B}$, which is a Choi state of some channel $\Psi$, and let $\ttt{B}=(M_0,B,\{P^k_m\}_{k,m},\Psi)$ be a $1$-to-$1$ $N$-message cloning attack against $\ttt{Q}$. Since $\Psi$ is an entanglement-breaking channel, $\mfk{c}^{N}_{1\rightarrow 1}(\ttt{Q}|\ttt{B})\leq\mfk{c}^{N}_{1\rightarrow 1}(\ttt{Q}|\scr{M})$. On the other hand, let $\tilde{\ttt{A}}=(M_0,B,\{P^k_m\},\Phi)$ be a $1$-to-$1$ $N$-message cloning attack against $\ttt{Q}$. By construction $\mfk{c}^{N}_{1\rightarrow s}(\ttt{Q}|\ttt{A})\leq\mfk{c}^{N}_{1\rightarrow 1}(\ttt{Q}|\tilde{\ttt{A}})$, and we also have
\begin{align*}
    \abs*{\mfk{c}^N_{1\rightarrow 1}(\ttt{Q}|\tilde{\ttt{A}})-\mfk{c}^{N}_{1\rightarrow 1}(\ttt{Q}|\ttt{B})}&\leq\int_K\frac{1}{N}\sum_{m\in M_0}\abs*{\Tr\squ{(d\bar{\sigma}^k_m\otimes P^k_m)\rho_{AB}}-\Tr\squ{(d\bar{\sigma}^k_m\otimes P^k_m)\tau_{AB}}}d\pi(k)\\
    &\leq\int_K\frac{1}{N}\sum_{m\in M_0}\sum_ip_i\abs*{\Tr\squ{(d\bar{\sigma}^k_m\otimes P^k_m)\rho_i}-\Tr\squ{(d\bar{\sigma}^k_m\otimes P^k_m)(\rho_{i,A}\otimes\rho_{i,B})}}d\pi(k)\\
    &=\int_K\frac{1}{N}\sum_ip_i\sum_{m\in M_0}p_{ikm}\abs*{\Tr\squ{d\bar{\sigma}^k_m\rho_{ikm}}-\Tr\squ{d\bar{\sigma}^k_m\rho_{i,A}}}d\pi(k)\\
    &=\sum_ip_i\int_K\frac{d}{N}\sum_{m\in M_0}p_{ikm}\norm{\sigma^k_m}\norm{\rho_{ikm}-\rho_{i,A}}_{\Tr}d\pi(k)\\
    &\leq\sum_{i\in I}p_i(2\eta d)+\sum_{i\notin I}p_i\frac{\eta d}{N}\sqrt{\frac{\varepsilon}{2}}\\
    &\leq\frac{2\eta d\log d}{s\varepsilon}+\frac{\eta d}{N}\sqrt{\frac{\varepsilon}{2}}.
\end{align*}
Taking $\varepsilon=\parens*{4\sqrt{2}\frac{N\log d}{s}}^{2/3}$ gives the upper bound $3\eta d\parens*{\frac{\log d}{N^2s}}^{1/3}$. Hence, we get that
\begin{align*}
    \mfk{c}^{N}_{1\rightarrow 1}(\ttt{Q}|\scr{M})\geq\mfk{c}^{N}_{1\rightarrow 1}(\ttt{Q}|\ttt{B})\geq\mfk{c}^{N}_{1\rightarrow 1}(\ttt{Q}|\tilde{\ttt{A}})-3\eta d\parens*{\frac{\log d}{N^2s}}^{1/3}\geq\mfk{c}^{N}_{1\rightarrow s}(\ttt{Q}|\ttt{A})-3\eta d\parens*{\frac{\log d}{N^2s}}^{1/3}.
\end{align*}
Taking the supremum over attacks $\ttt{A}$ gives the wanted result.
\end{proof}
\section{Lower bounds}\label{sec:lowerbounds}

In this section, we study lower bounds on telegraphing attacks against the Haar measure encryption.

\subsection{Encryption of a bit}

Consider the following telegraphing attack for the rank-$r$ Haar-measure encryption of a bit: $\ttt{A}=(B,\{P^k_m\},\Phi)$, where $B=H=\C^{[2r]}$, $\Phi$ is measurement in the computational basis $\Phi(\rho)=\sum_{i=0}^{2r-1}\ketbra{i}\rho\ketbra{i}$, and let $P_{b}^U=\sum_{i\in G_{b,U}}\ketbra{i}{i}$, where $G_{0,U}=\set*{i\in[2r]}{\braket{i}{U\sigma_0U^\ast|i}\geq\braket{i}{U\sigma_{1}U^\ast|i}}$ and $G_{1,U}=[2r]\backslash G_{0,U}$.

\begin{proposition}\label{prop:lowerbound}
    The winning probability of the attack $\ttt{A}$ above is
    \begin{align*}
        \mfk{c}_{1\rightarrow 1}(\ttt{Q}_{r,2}|\ttt{A})=\frac{1}{2}+\frac{1}{2^{2r+1}}\binom{d}{r}.
    \end{align*}
\end{proposition}

\begin{lemma}
    Let $f:S^{n-1}\rightarrow\RR$ be a function. Then
    \begin{align*}
        &\int_{\RR^n}f\parens*{\frac{v}{\norm{v}}}e^{-\norm{v}^2}d^nv=\frac{1}{2}\Gamma(n/2)\int_{S^{n-1}}f(\Omega)d^{n-1}\Omega\\
        &\int_{\RR^n}f\parens*{\frac{v}{\norm{v}}}e^{-\norm{v}^2}\norm{v}^2d^nv=\frac{n}{4}\Gamma(n/2)\int_{S^{n-1}}f(\Omega)d^{n-1}\Omega,
    \end{align*}
    where the integrations are with respect to the unnormalised Lebesgue measures.
\end{lemma}

\begin{proof}
    We can express the integrations in spherical coordinates, that is the change of variables $v=r\Omega$, where $r\geq 0$ and $\Omega$ is a unit vector to get $d^nv=r^{n-1}drd^{n-1}\Omega$. As such, we get in the first case that
    \begin{align*}
        \int_{\RR^n}f\parens*{\frac{v}{\norm{v}}}e^{-\norm{v}^2}d^nv&=\int_{S^{n-1}}\int_0^\infty f(\Omega)e^{-r^2}r^{n-1}drd^{n-1}\Omega.
    \end{align*}
    With the change of variables $t=r^2$, $\int_0^\infty e^{-r^2}r^{n-1}dr=\frac{1}{2}\int_0^\infty t^{n/2-1}e^{-t}dt=\frac{\Gamma(n/2)}{2}$, as wanted. The other case is similar. There,
    \begin{align*}
        \int_{\RR^n}f\parens*{\frac{v}{\norm{v}}}e^{-\norm{v}^2}\norm{v}^2d^nv&=\int_{S^{n-1}}\int_0^\infty f(\Omega)e^{-r^2}r^{n+1}drd^{n-1}\Omega,
    \end{align*}
    and with the same change of variables $\int_0^\infty e^{-r^2}r^{n+1}dr=\frac{1}{2}\int_0^\infty t^{n/2}e^{-t}dt=\frac{\Gamma(n/2+1)}{2}=\frac{n}{4}\Gamma(n/2)$.
\end{proof}

\begin{proof}[Proof of \cref{prop:lowerbound}]
    The winning probability of the strategy is
    \begin{align*}
        \mfk{c}_{1\rightarrow 1}(\ttt{Q}_{r,2}|\ttt{A})&=\int\frac{1}{2}\sum_b\Tr\squ*{\Phi(U\tfrac{1}{r}\Pi_bU^\ast)P_{b|U}}dU\\
        &=\frac{1}{d}\int\sum_{b}\sum_{x\in G_{
       b,U}}\braket{x}{U\Pi_bU^\ast}{x}dU\\
       &=\frac{1}{d}\sum_{x}\int\max_b\braket{x}{U\Pi_bU^\ast}{x}dU\\
       &=\int\max_b\braket{0}{U\Pi_bU^\ast}{0}dU,
    \end{align*}
    by Haar invariance, where $d=2r$. Now, $U^\ast\ket{0}$ is just a uniformly random pure state, so this is just the normalised integral over the set of pure states $\mfk{c}_{1\rightarrow 1}(\ttt{Q}_{r,2}|\ttt{A})=\int\max_b\braket{\psi}{\Pi_b}{\psi}d\psi=\int\max_b\sum_{i=br}^{(b+1)r-1}|\psi_i|^2d\psi$. This is now equivalent to the integral over the real $2d-1$-sphere $\mfk{c}_{1\rightarrow 1}(\ttt{Q}_{r,2}|\ttt{A})=\frac{\Gamma(d)}{2\pi^d}\int_{S^{2d-1}}\max_b\sum_{i=bd}^{(b+1)d-1}\Omega_i^2d^{2d-1}\Omega$, since $\frac{2\pi^d}{\Gamma(d)}$ is the volume of the $2d-1$-sphere. Using the lemma,
    \begin{align*}
        \mfk{c}_{1\rightarrow 1}(\ttt{Q}_{r,2}|\ttt{A})&=\frac{\Gamma(d)}{2\pi^d}\frac{2}{d\Gamma(d)}\int_{\RR^{2d}}\max_b\sum_{i=bd}^{(b+1)d-1}\frac{v_i^2}{\norm{v}^2}e^{-\norm{v}^2}\norm{v}^2d^{2d}v\\
        &=\frac{1}{\pi^dd}\int_{\RR^{2d}}\max_b\sum_{i=bd}^{(b+1)d-1}v_i^2e^{-\norm{v}^2}d^{2d}v
    \end{align*}
    Now, let the vectors $v^0=(v_0,\ldots,v_{d-1})$ and $v^1=(v_{d},\ldots,v_{2d-1})$. Then, using a change of variables to spherical coordinates
    \begin{align*}
        \mfk{c}_{1\rightarrow 1}(\ttt{Q}_{r,2}|\ttt{A})&=\frac{1}{\pi^dd}\int_{\RR^{d}}\int_{\RR^d}\max_b\norm{v^b}^2e^{-\norm{v^0}^2-\norm{v^1}^2}d^{d}v^0d^dv^1\\
        &=\frac{1}{\pi^dd}\int_{S^{d-1}}\int_0^\infty\int_{S^{d-1}}\int_0^\infty\max_b(r^b)^2e^{-(r^0)^2-(r^1)^2}(r^0)^{d-1}dr^0d^{d-1}\Omega^0(r^1)^{d-1}dr^1d^{d-1}\Omega^1\\
        &=\frac{1}{\pi^dd}\parens*{\frac{2\pi^{d/2}}{\Gamma(d/2)}}^2\int_0^\infty\int_0^\infty\max_b(r^b)^2e^{-(r^0)^2-(r^1)^2}(r^0)^{d-1}dr^0(r^1)^{d-1}dr^1
        \\
        &=\frac{4}{d(d/2-1)!^2}\int_0^\infty\int_0^\infty\max_b(r^b)^2e^{-(r^0)^2-(r^1)^2}(r^0)^{d-1}(r^1)^{d-1}dr^0dr^1.
    \end{align*}
    Next, take the change of variables $s=(r^0)^2$ and $t=(r^1)^2$. We find that
    \begin{align*}
        \mfk{c}_{1\rightarrow 1}(\ttt{Q}_{r,2}|\ttt{A})&=\frac{1}{d(d/2-1)!^2}\int_0^\infty\int_0^\infty\max\{s,t\}e^{-s-t}s^{d/2-1}t^{d/2-1}dsdt\\
        &=\frac{2}{d(d/2-1)!^2}\int_0^\infty\int_0^te^{-s-t}s^{d/2-1}t^{d/2}dsdt.
    \end{align*}
    Now, $\int_0^ts^{d/2-1}e^{-s}ds=\gamma(d/2,t)=(d/2-1)!-(d/2-1)!e^{-t}\sum_{q=0}^{d/2-1}\frac{t^q}{q!}$, the lower incomplete gamma function. Therefore,
    \begin{align*}
        \mfk{c}_{1\rightarrow 1}(\ttt{Q}_{r,2}|\ttt{A})&=\frac{2}{d(d/2-1)!}\int_0^\infty t^{d/2}e^{-t}(d/2-1)!\parens*{1-e^{-t}\sum_{q=0}^{d/2-1}\frac{t^q}{q!}}dt\\
        &=\frac{1}{(d/2)!}\int_0^\infty t^{d/2}e^{-t}-\sum_{q=0}^{d/2-1}\frac{1}{q!}t^{d/2+q}e^{-2t}dt\\
        &=\frac{1}{(d/2)!}\parens*{(d/2)!-\sum_{q=0}^{d/2-1}\frac{1}{q!}\frac{(d/2+q)!}{2^{d/2+q+1}}}\\
        &=1-\sum_{q=0}^{d/2-1}\binom{d/2+q}{q}\frac{1}{2^{q+d/2+1}}.
    \end{align*}
    To simplify this formula, the Beta function is $\mathrm{B}(a,b)=\int_0^1t^{a-1}(1-t)^{b-1}dt$ and the regularised incomplete Beta function is $I_p(a,b)=\frac{1}{\mathrm{B}(a,b)}\int_0^pt^{a-1}(1-t)^{b-1}dt$. $I_p(a,b)$ satisfies the relations $I_p(a,b)=I_{1-p}(b,a)$, and for integers $m,n$, $I_p(m,n-m+1)=\sum_{j=m}^n\binom{n}{j}p^j(1-p)^{n-j}$ and $I_p(m,n)=\sum_{j=m}^\infty\binom{n+j-1}{j}p^j(1-p)^n$~\cite[\href{https://dlmf.nist.gov/8.17}{(8.17)}]{DLMF}. Also, $I_p(1,b)=1-(1-p)^b$, giving that
    \begin{align*}
        \sum_{q=0}^{d/2-1}\binom{d/2+q}{q}\frac{1}{2^{q+d/2+1}}=\frac{1}{2^{d/2+1}}+I_{1/2}(1,d/2+1)-I_{1/2}(d/2,d/2+1)=1-I_{1/2}(d/2,d/2+1).
    \end{align*}
    Hence,
    \begin{align*}
        \mfk{c}_{1\rightarrow 1}(\ttt{Q}_{r,2}|\ttt{A})=I_{1/2}(d/2,d/2+1)=\sum_{j=d/2}^d\binom{d}{j}\frac{1}{2^d}.
    \end{align*}
    We have that $\sum_{j=0}^d\binom{d}{j}=2^d$ and $\sum_{j=d/2}^d\binom{d}{j}=\sum_{j=0}^{d/2}\binom{d}{j}$. Hence, $2^d=\sum_{j=d/2}^d\binom{d}{j}+\sum_{j=1}^{d/2-1}\binom{d}{j}=2\sum_{j=d/2}^d\binom{d}{j}-\binom{d}{d/2}$, and hence
    \begin{align*}
        \mfk{c}_{1\rightarrow 1}(\ttt{Q}_{r,2}|\ttt{A})&=I_{1/2}(d/2,d/2+1)=\frac{1}{2^d}\parens*{2^{d-1}+\frac{1}{2}\binom{d}{d/2}}=\frac{1}{2}+\frac{1}{2^{d+1}}\binom{d}{d/2}. \qedhere
    \end{align*}
\end{proof}

\begin{corollary}\label{cor:better-lowerbound}
    $\mfk{c}_{1\rightarrow 1}(\ttt{Q}_{r,2}|\ttt{A})=\frac{1}{2}+\frac{1}{\sqrt{2\pi d}}+O\parens*{\frac{1}{d^{3/2}}}$.
\end{corollary}

\begin{proof}
    By Stirling's approximation $\sqrt{2\pi n}\parens*{\frac{n}{e}}^ne^{\frac{1}{12n+1}}\leq n!\leq\sqrt{2\pi n}\parens*{\frac{n}{e}}^ne^{\frac{1}{12n}}$~\cite{Rob55}, so
	\begin{align*}
		\binom{2n}{n}=\frac{(2n)!}{(n!)^2}\leq\frac{\sqrt{4\pi n}\parens*{\frac{2n}{e}}^{2n}e^{\frac{1}{24n}}}{2\pi n\parens*{\frac{n}{e}}^{2n}e^{\frac{2}{12n+1}}}=\frac{4^n}{\sqrt{\pi n}}e^{\frac{1}{24n}-\frac{2}{12n+1}}\leq\frac{4^n}{\sqrt{\pi n}},
	\end{align*}
    and
    \begin{align*}
		\binom{2n}{n}=\frac{(2n)!}{(n!)^2}\geq\frac{\sqrt{4\pi n}\parens*{\frac{2n}{e}}^{2n}e^{\frac{1}{24n+1}}}{2\pi n\parens*{\frac{n}{e}}^{2n}e^{\frac{2}{12n}}}=\frac{4^n}{\sqrt{\pi n}}e^{\frac{1}{24n+1}-\frac{2}{12n}}\geq\frac{4^n}{\sqrt{\pi n}}\parens*{1-\frac{1}{6n}}.
	\end{align*}
    Therefore, $\frac{1}{2}+\frac{1}{\sqrt{2\pi d}}\parens*{1-\frac{1}{3d}}\leq \mfk{c}_{1\rightarrow 1}(\ttt{Q}_{r,2}|\ttt{A})\leq\frac{1}{2}+\frac{1}{\sqrt{2\pi d}}$, giving the wanted asymptotic formula.
\end{proof}

\subsection{Extension to more copies}

To extend to $t$ copies, we can apply the same strategy $\ttt{A}$ of the previous section independently $t$ times by measuring in a uniformly random basis for each copy. Then, the player chooses their output by choosing the most common bit, guessing uniformly random if $0$ and $1$ are equally common. Formally, $\ttt{A}^{(t)}=(B^{(t)},\{P^{(t),U}_{m}\},\Phi^{(t)})$, where $B^{(t)}=B^{\otimes t}$, $\Phi^{(t)}=\Phi^{\otimes t}$, and $P^{(t),U}_{m}=\sum_{\substack{m_1,\ldots,m_t\\\mathrm{MAJ}(m_1,\ldots,m_t)=m}}P^U_{m_1}\otimes\cdots\otimes P^U_{m_t}$. Then, if $p$ is the winning probability of one round,
\begin{align*}
    \mfk{c}_{t\rightarrow 1}(\ttt{Q}_{r,2}|\ttt{A}^{(t)})=\begin{cases}\sum_{\ell=\frac{t+1}{2}}^t\binom{t}{\ell}p^\ell(1-p)^{t-\ell}&t\text{ odd}\\\sum_{\ell=\frac{t}{2}+1}^t\binom{t}{\ell}p^\ell(1-p)^{t-\ell}+\frac{1}{2}\binom{t}{t/2}(p(1-p))^{t/2}&t\text{ even}\end{cases}
\end{align*}

\begin{lemma}\label{lem:t-copy-lowerbound}
    Let $p=\frac{1}{2}+\delta$. Suppose that $t\geq 4$ and $\delta\leq\frac{1}{2\sqrt{t-1}}$. Then, $\frac{1}{2}+\frac{1}{3}\sqrt{t}\delta\leq \mfk{c}_{t\rightarrow 1}(\ttt{Q}_{r,2}|\ttt{A}^{(t)})\leq\frac{1}{2}+\sqrt{t}\delta$.
\end{lemma}

\begin{proof}
    First, consider the case where $t$ is odd. We will again use properties of the regularised incomplete beta function~\cite[\href{https://dlmf.nist.gov/8.17}{(8.17)}]{DLMF}. First, $\mfk{c}_{t\rightarrow 1}(\ttt{Q}_{r,2}|\ttt{A}^{(t)})=I_p(\frac{t+1}{2},\frac{t+1}{2})$. Using the property $I_x(a,a)=\frac{1}{2}I_{4x(1-x)}(a,1/2)$ for $x\in[0,1/2]$, we find that
    \begin{align*}
        \mfk{c}_{t\rightarrow 1}(\ttt{Q}_{r,2}|\ttt{A}^{(t)})&=1-I_{1-p}(\tfrac{t+1}{2},\tfrac{t+1}{2})\\
        &1-\tfrac{1}{2}I_{4p(1-p)}(\tfrac{t+1}{2},\tfrac{1}{2})\\
        &=\tfrac{1}{2}+\tfrac{1}{2}I_{1-4p(1-p)}(\tfrac{1}{2},\tfrac{t+1}{2}).
    \end{align*}
    Now, note that $1-4p(1-p)=4\delta^2$, so
    \begin{align*}
        \mfk{c}_{t\rightarrow 1}(\ttt{Q}_{r,2}|\ttt{A}^{(t)})&=\frac{1}{2}+\frac{1}{2}\frac{\Gamma(\frac{t+1}{2}+\frac{1}{2})}{\Gamma(\frac{t+1}{2})\Gamma(\frac{1}{2})}\int_0^{4\delta^2}t^{-\frac{1}{2}}(1-t)^{\frac{t-1}{2}}dt\\
        &\geq\frac{1}{2}+\frac{1}{2}\frac{\frac{(t+1)!}{2^{t+1}\parens*{\frac{t+1}{2}}!}\sqrt{\pi}}{\parens*{\frac{t-1}{2}}!\sqrt{\pi}}\int_0^{4\delta^2}t^{-\frac{1}{2}}(1-4\delta^2)^{\frac{t-1}{2}}dt\\
        &=\frac{1}{2}+\frac{1}{2^{t+2}}\frac{t+1}{2}\binom{t+1}{\frac{t+1}{2}}(1-4\delta^2)^{\frac{t-1}{2}}2\sqrt{4\delta^2}\\
        &\geq\frac{1}{2}+(1-4\delta^2)^{\frac{t-1}{2}}(t+1)\frac{1}{\sqrt{\pi\frac{t+1}{2}}}\parens*{1-\frac{1}{3(t+1)}}\delta\\
        &\geq\frac{1}{2}+\frac{1}{2}\sqrt{\frac{2}{\pi}}\frac{14}{15}\sqrt{t+1}\delta\\
        &\geq\frac{1}{2}+\frac{1}{3}\sqrt{t}\delta.
    \end{align*}
    In the same way we can upper bound
    
    \begin{align*}
        \mfk{c}_{t\rightarrow 1}(\ttt{Q}_{r,2}|\ttt{A}^{(t)})&\leq\frac{1}{2}+\frac{1}{2^{t+2}}\frac{t+1}{2}\binom{t+1}{\frac{t+1}{2}}\int_0^{4\delta^2}t^{-\frac{1}{2}}dt\\
        &\leq\frac{1}{2}+\frac{1}{2}\frac{t+1}{2}\frac{1}{\sqrt{\pi\frac{t+1}{2}}}2\sqrt{4\delta^2}\\
        &=\frac{1}{2}+\sqrt{\frac{2}{\pi}}\sqrt{t+1}\delta\\
        &\leq\frac{1}{2}+\sqrt{t}\delta
    \end{align*}
    Now, in the case that $t$ is even, \begin{align*}
        \mfk{c}_{t\rightarrow 1}(\ttt{Q}_{r,2}|\ttt{A}^{(t)})&=I_p(\tfrac{t}{2}+1,\tfrac{t}{2})+\frac{1}{2}\binom{t}{t/2}(p(1-p))^{t/2}\\
        &=I_p(\tfrac{t}{2},\tfrac{t}{2})-\frac{(p(1-p))^{t/2}\Gamma(t)}{\frac{t}{2}\Gamma(t/2)^2}+\frac{1}{2}\binom{t}{t/2}(p(1-p))^{t/2}\\
        &=I_p(\tfrac{t}{2},\tfrac{t}{2})\\
        &\in\squ*{\frac{1}{2}+\frac{1}{3}\sqrt{t}\delta,\frac{1}{2}+\sqrt{t}\delta},
    \end{align*}
    by the above calculation with $t+1$ replaced by $t$.
\end{proof}

\subsection{Extension to more messages}

We can use essentially the same strategy for the Haar-random encryption of longer messages. Let $\ttt{A}=(B,\{P^U_m\},\Phi)$ be the telegraphing attack against $\ttt{Q}_{r,n}$ where $B=H=\C^{rn}$, $\Phi$ is measurement in the computational basis, and $P_{m}^U=\sum_{i\in G_{m,U}}\ketbra{i}$ where $$G_{m,U}=\set*{i}{\forall m'< m.\;\braket{i}{U\sigma_m U^\ast}{i}>\braket{i}{U\sigma_{m'} U^\ast}{i}\land\forall m'>m.\;\braket{i}{U\sigma_m U^\ast}{i}\geq\braket{i}{U\sigma_{m'} U^\ast}{i}}.$$

\begin{proposition}
    The winning probability for the strategy $\ttt{A}$ above is
    $$\mfk{c}_{1\rightarrow 1}(\ttt{Q}_{r,n}|\ttt{A})=\sum_{q_1,\ldots,q_{n-1}=r}^\infty\binom{q_1+\ldots+q_{n-1}+r}{q_1,\ldots,q_{n-1},r}n^{-(q_1+\ldots+q_{n-1}+r+1)}.$$
\end{proposition}

\begin{proof}
    The winning probability can be simplified to
    \begin{align*}
        \mfk{c}_{1\rightarrow 1}(\ttt{Q}_{r,n}|\ttt{A})&=\int\frac{1}{n}\sum_m\Tr\squ*{\Phi(U\tfrac{1}{r}\Pi_m U^\ast)P_{m}^U}dU\\
        &=\frac{1}{d}\int\sum_m\sum_{i\in G_{m,U}}\braket{i}{U\Pi_mU^\ast}{i}dU\\
        &=\int\max_m\braket{0}{U\Pi_mU^\ast}{0}dU,
    \end{align*}
    where $d=rn$ and $\Pi_m=\ketbra{m}\otimes I_r=r\sigma_m$. As in the one-bit case, $U^\ast\ket{0}$ is just a uniformly random pure state, so the winning probability can be expressed as an integral over the real $2d-1$-sphere:
    \begin{align*}
        \mfk{c}_{1\rightarrow 1}(\ttt{Q}_{r,n}|\ttt{A})&=\frac{\Gamma(d)}{2\pi^d}\int_{S^{2d-1}}\max_m\sum_{j=2rm}^{2r(m+1)-1}\Omega_j^2d^{2d-1}\Omega\\
        &=\frac{1}{\pi^dd}\int_{\RR^{2d}}\max_m\sum_{j=2rm}^{2r(m+1)-1}|v_j|^2e^{-\norm{v}^2}d^{2d}v.
    \end{align*}
    We rewrite the integration as an integration over $n$ vectors $v^m=(v_{2rm+1},\ldots,v_{2r(m+1)})$, and then change to spherical coordinates:
    \begin{align*}
        \mfk{c}_{1\rightarrow 1}(\ttt{Q}_{r,n}|\ttt{A})&=\frac{1}{\pi^dd}\int_{\RR^{2r}}\cdots\int_{\RR^{2r}}\max_m\norm{v^m}^2e^{-\sum_m\norm{v^m}^2}dv^1\cdots dv^n\\
        &=\frac{1}{\pi^dd}\parens*{\frac{2\pi^{r}}{\Gamma(r)}}^{m}\int_0^\infty\cdots\int_0^\infty\max_m(r^m)^2e^{-\sum_{m'}(r^{m'})^2}(r^1)^{2r-1}dr^1\cdots(r^n)^{2r-1}dr^n\\
        &=\frac{1}{d(r-1)!^m}\int_0^\infty\cdots\int_0^\infty\max_mt_me^{-\sum_{m'}t_{m'}}t_1^{r-1}dt_1\cdots t_n^{r-1}dt_n,
    \end{align*}
    via the change of variables $t_m=(r^m)^2$. Now, we can split the integration over $n$ subsets $D_m=\set*{(t_1,\ldots,t_n)}{t_m\geq t_{m'}\forall m'}$. By symmetry, we have that
    \begin{align*}
        \mfk{c}_{1\rightarrow 1}(\ttt{Q}_{r,n}|\ttt{A})&=\frac{m}{d(r-1)!^m}\int_0^\infty\int_0^{t_n}\cdots\int_0^{t_n}t_ne^{-\sum_mt_m}t_1^{r-1}dt_1\cdots t_m^{r-1}dt_n\\
        &=\frac{1}{r!(r-1)!^{m-1}}\int_0^\infty\parens*{\int_0^ts^{r-1}e^{-s}ds}^{n-1}t^{r}e^{-t}dt.
    \end{align*}
    Again, $\int_0^ts^{r-1}e^{-s}ds=\gamma(r,t)=(r-1)!e^{-t}\sum_{q=r}^\infty\frac{t^q}{q!}$, so
    \begin{align*}
        \mfk{c}_{1\rightarrow 1}(\ttt{Q}_{r,n}|\ttt{A})&=\frac{1}{r!}\int_0^\infty\parens*{e^-t\sum_{q=r}^\infty\frac{t^q}{q!}}^{n-1}t^{r}e^{-t}dt\\
        &=\frac{1}{r!}\sum_{q_1,\ldots,q_{n-1}=r}^\infty\frac{1}{q_1!\cdots q_{n-1}!}\int_0^\infty t^{q_1+\ldots+q_{n-1}+r}e^{-nt}dt\\
        &=\sum_{q_1,\ldots,q_{n-1}=r}^\infty\frac{(q_1+\ldots+q_{n-1}+r)!}{q_1!\cdots q_{n-1}!r!}n^{-(q_1+\ldots+q_{n-1}+r+1)}.\qedhere
    \end{align*}
\end{proof}

\subsection{Untelegraphable-indistinguishable security}

We can study a telegraphing-distinguishing attack similar to the telegraphing attack of the previous section, which allows us to find a more concrete lower bound on the Haar measure encryption, with a closed-form expression. Consider the telegraphing-distinguishing attack $\ttt{A}=(\{m_0,m_1\},B,\{P^U_b\},\Phi)$ against the Haar-measure encryption of $n$ messages, where $m_0,m_1$ are arbitrary distinct messages, $B=H=\C^{[rn]}$, $\Phi$ is measurement in the computational basis, and $P_{b}^U=\sum_{i\in G_{b,U}}\ketbra{i}{i}$, where $G_{0,U}=\set*{i}{\braket{i}{U\sigma_{m_0}U^\ast|i}\geq\braket{i}{U\sigma_{m_1}U^\ast|i}}$ and $G_{1,U}=[2r]\backslash G_{0,U}$.

\begin{proposition}\label{prop:td-lowerbound}
    The telegraphing-distinguishing probability of the attack $\ttt{A}$ above is
    $$\mfk{c}_{1\rightarrow 1}^2(\ttt{Q}_{r,n}|\ttt{A})=\frac{1}{2}+\frac{1}{2^{2r+1}}\binom{2r}{r}=\frac{1}{2}+\frac{1}{2\sqrt{\pi r}}+O\parens*{\frac{1}{r^{3/2}}}.$$
\end{proposition}

\begin{proof}
    The proof begins identically to the one-bit case, but with $d/2$ replaced with $r$. First, 
    \begin{align*}
        \mfk{c}_{1\rightarrow 1}^2(\ttt{Q}_{r,n}|\ttt{A})&=\int\frac{1}{2}\sum_b\Tr\squ*{\Phi(U\tfrac{1}{r}\Pi_{m_b}U^\ast)P_{b}^U}dU\\
        &=\frac{1}{2r}\int\sum_{b}\sum_{i\in G_{
       b,U}}\braket{i}{U\Pi_{m_b}U^\ast}{i}dU\\
       &=\frac{n}{2}\int\max_b\braket{0}{U\Pi_{m_b}U^\ast}{0}dU\\
       &=\frac{n\Gamma(d)}{4\pi^d}\int_{S^{2d-1}}\max_b\sum_{i=2br}^{2(b+1)r-1}\Omega_i^2d^{2d-1}\Omega\\
       &=\frac{1}{2r\pi^d}\int_{\RR^{2d}}\max_b\sum_{i=2br}^{2(b+1)r-1}v_i^2e^{-\norm{v}^2}d^{2d}v.
    \end{align*}
    Now, we rewrite the integration as an integration over $3$ vectors $v^0=(v_{0},\ldots,v_{2r-1})$, $v^1=(v_{2r},\ldots,v_{4r-1})$, $v^2=(v_{4r},\ldots,v_{2d-1})$, and then integrate $v^2$:
    \begin{align*}
        w_{dist}(\Phi,P)&=\frac{1}{2r\pi^d}\int_{\RR^{2d-4r}}\int_{\RR^{2r}}\int_{\RR^{2r}}\max_{b=0,1}\norm{v^b}^2e^{-\norm{v^0}^2-\norm{v^1}^2-\norm{v^2}^2}d^{2r}v^0d^{2r}v^1d^{2d-4r}v^2\\
        &=\frac{1}{2r\pi^d}\pi^{d-2r}\int_{\RR^{2r}}\int_{\RR^{2r}}\max_{b=0,1}\norm{v^b}^2e^{-\norm{v^0}^2-\norm{v^1}^2}d^{2r}v^0d^{2r}v^1.
    \end{align*}
    This is equal to the winning probability of the one-bit game with $d=2r$, giving the result by~\cref{prop:lowerbound,cor:better-lowerbound}.
\end{proof}

\begin{corollary}\label{cor:lower-from-dist-teleg}
    The telegraphing value of the rank-$r$ Haar-measure encryption of $n$ messages is lower-bounded as
    $$\mfk{c}_{1\rightarrow 1}(\ttt{Q}_{r,n}|\scr{M})\geq \frac{1}{n}+\frac{1}{n\sqrt{\pi r}}+O\parens*{\frac{1}{nr^{3/2}}}$$
\end{corollary}

This follows by applying the argument of~\cite[Theorem 12]{BL20}. Essentially, to adapt the attack to a telegraphing attack, it proceeds in the same way and on output $b$ from the telegraphing-distinguishing attack, the adversary outputs $m_b$. Then, if the original message is not $m_0$ or $m_1$, the attack always loses; but if the original message was $m_0$ or $m_1$, then the attack succeeds with the same probability as the telegraphing-distinguishing attack. As such, the new attack wins with probability $\frac{2}{n}$ times the telegraphing-distinguishing probability.

To finish this section, note we can combine this theorem with the results of \cref{sec:haar_random} to get upper and lower bounds on the telegraphing-distinguishing value that are tight in the order of $r$.

\begin{corollary} \label{cor:Haar_measure_scheme_lower_bound}
    The $t$-copy telegraphing-distinguishing value of the rank-$r$ Haar-measure encryption of $n$ messages is bounded as
    \begin{align*}
        \frac{1}{2}+\frac{1}{6}\sqrt{\frac{t}{\pi r}}+O\parens*{\frac{t^{1/2}}{r^{3/2}}}\leq\mfk{c}_{t\rightarrow 1}^2(\ttt{Q}_{r,n}|\scr{M})\leq\frac{1}{2}+\frac{7t}{\sqrt{r}}.
    \end{align*}
\end{corollary}

This corollary follows by combining \cref{prop:td-lowerbound,lem:t-copy-lowerbound,thm:haar_random_encryption_scheme_t_copy_untelegraphable_indistinguishable}
\section{Minimality of the Haar measure game} \label{sec:minimality}

\subsection{One-copy minimality}

In this section, we extend the minimality result of~\cite{MST21arxiv} to the context of general cloning attacks, which will allow us to get minimality for telegraphing attacks as well as cloning attacks to an arbitrary number of receivers.

\begin{theorem}\label{thm:haar-minimality}
    Let $\ttt{Q}=(M,K,\pi,H,\{\sigma^k_m\}_{k,m})$ be a correct QECM and let $\scr{F}$ be a class of channels that is closed under $\Phi\mapsto\Phi(V\cdot V^\ast)$ for every isometry $V$. Then for $r=\dim H-|M|+1$ and $n=|M|$, $\mfk{c}^N_{1\rightarrow s}(\ttt{Q}_{r,n}|\scr{F})\leq\mfk{c}^N_{1\rightarrow s}(\ttt{Q}|\scr{F})$.
\end{theorem}

For example, the classes of all channels, of measurement channels, and of entanglement-breaking channels satisfy the conditions of the theorem.

\begin{proof}
    We can assume without loss of generality that $M=[n]$. We proceed via a sequence of hybrid QECMs $\ttt{Q}^{(i)}$. First, we enlarge the dimension of the ciphertext. Let $V:H\rightarrow\C^{rn}$ be an isometry and let $\tilde{\sigma}^k_m=V\sigma^k_{m}V^\ast$. Since $\scr{F}$ is closed under preconjugation by isometries, it is clear $\mfk{c}^N_{1\rightarrow s}(\ttt{Q}^{(1)}|\scr{F})\leq\mfk{c}^N_{1\rightarrow s}(\ttt{Q}|\scr{F})$.

    Next, note that as $\ttt{Q}^{(1)}$ is correct, $\sum_m\rank(\tilde{\sigma}^k_m)=\sum_m\rank(\sigma^k_m)\leq \dim H$, so since $\rank(\tilde{\sigma}^k_m)\geq 1$, $$\rank(\tilde{\sigma}^k_m)\leq\dim H-\sum_{m'\neq m}\rank(\tilde{\sigma}^k_{m'})\leq\dim H-(n-1)=r.$$ Knowing this, we can write $\tilde{\sigma}^k_m=U_k\delta^k_m U_k^\ast$, where $U_k$ is unitary and $\delta^k_m$ is diagonal with support contained in $\mathrm{span}\set{\ket{i+r(m-1)}}{i\in[r]}$. Then, $\pi$ induces a distribution $\pi'$ on $\mc{U}(rn)\times D_{rn}^n$, where $D_{rn}\subseteq\mc{D}(\C^{rn})$ is the set of diagonal density matrices, as $$\int_{\mc{U}(rn)\times D_{rn}^n}f(U,\delta_1,\ldots,\delta_n)d\pi'(U,\delta_1,\ldots,\delta_n)=\int_kf(U_k,\delta^k_1,\ldots,\delta^k_n)d\pi(k).$$
    Let $\mu$ be the marginal of $\pi'$ on $D_{rn}^n$, let $\pi'_{\delta_1,\ldots,\delta_n}$ be the conditional distribution of $\pi'$ on $\mc{U}(rn)$ given $\delta_1,\ldots,\delta_n$, and define $\sigma^{U,\delta_1,\ldots,\delta_n}_m=U\delta_m U^\ast$. Define the QECM $\ttt{Q}^{(2)}=([n],\mc{U}(rn)\times D_{rn}^n,\mu_{Haar}\times\mu,\C^{rn},\{\sigma^{U,\delta_1,\ldots,\delta_n}_m\}_{(U,\delta_1,\ldots,\delta_n),m})$. This is again a correct QECM. Let $\ttt{A}=(M_0,\{B_i\}_i,\{P^{i,U,\delta_1,\ldots,\delta_n}_m\}_{i,U,\delta_1,\ldots,\delta_n,m},\Phi)$ be a $1$-to-$s$ $N$-message cloning attack over $\scr{F}$ against $\ttt{Q}^{(2)}$. For each $U\in\mc{U}(rn)$, define the $1$-to-$s$ $N$-message cloning attack over $\scr{F}$ against $\ttt{Q}^{(1)}$ as $\ttt{A}_{U}=(M_0,\{B_i\},\{P^{i,UU_k,\delta^k_1,\ldots,\delta^k_n}_{m}\}_{i,k,m},\Phi_U)$, where $\Phi_U(\rho)=\Phi(U\rho U^\ast)$. Using Haar invariance, we find that
    \begin{align*}
        \mfk{c}^N_{1\rightarrow s}&(\ttt{Q}^{(2)}|\ttt{A})=\int_{\mc{U}(rn)}\int_{D_{rn}^n}\frac{1}{N}\sum_{m\in M_0}\Tr\squ{(P^{1,U,\delta_1,\ldots,\delta_n}_m\otimes\cdots\otimes P^{s,U,\delta_1,\ldots,\delta_n}_m)\Phi(U\delta_m U^\ast)}d\mu(\delta_1,\ldots,\delta_n)dU\\
        &=\int_{\mc{U}(rn)}\int_{D_{rn}^n}\int_{\mc{U}(rn)}\frac{1}{N}\sum_{m\in M_0}\Tr\squ{(P^{1,U,\delta_1,\ldots,\delta_n}_m\otimes\cdots\otimes P^{s,U,\delta_1,\ldots,\delta_n}_m)\Phi(U\delta_m U^\ast)}d\pi'_{\delta_1,\ldots,\delta_n}(V)d\mu(\delta_1,\ldots,\delta_n)dU\\
        &=\int_{\mc{U}(rn)}\int_{\mc{U}(rn)\times D_{rn}^n}\frac{1}{N}\sum_{m\in M_0}\Tr\squ{(P^{1,UV,\delta_1,\ldots,\delta_n}_m\otimes\cdots\otimes P^{s,UV,\delta_1,\ldots,\delta_n}_m)\Phi(UV\delta_m V^\ast U^\ast)}d\pi'(V,\delta_1,\ldots,\delta_n)dU\\
        &=\int_{\mc{U}(rn)}\int_{K}\frac{1}{N}\sum_{m\in M_0}\Tr\squ{(P^{1,UU_k,\delta^k_1,\ldots,\delta^k_n}_m\otimes\cdots\otimes P^{s,UU_k,\delta_1^k,\ldots,\delta^k_n}_m)\Phi(UU_k\delta^k_m U_k^\ast U^\ast)}d\pi(k)dU\\
        &=\int_{\mc{U}(rn)}\mfk{c}^N_{1\rightarrow s}(\ttt{Q}^{(1)}|\ttt{A}_U)dU\leq\mfk{c}^N_{1\rightarrow s}(\ttt{Q}^{(1)}|\scr{F}),
    \end{align*}
    so $\mfk{c}^N_{1\rightarrow s}(\ttt{Q}^{(2)}|\scr{F})\leq\mfk{c}^N_{1\rightarrow s}(\ttt{Q}^{(1)}|\scr{F})$.

    Now, for any permutation $\tau$ of $[rn]$, write $V_\tau$ for the action on $\C^{rn}$ given by $V_\tau\ket{i}=\ket{\tau(i)}$. Let $S\subseteq \mfk{S}_{rn}$ be the set of permutations that preserves the intervals $[r(m-1)+1,\ldots,rm]$ for each $m\in [n]$. We have that $S\cong \mfk{S}_{r}^n$. Let $\upsilon$ be the uniform distribution $S$ and define the correct QECM $\ttt{Q}^{(3)}=([n],\mc{U}(rn)\times D_{rn}^n\times S,\mu_{Haar}\times\mu\times\upsilon,\C^{rn},\{\sigma_m^{UV_\tau,\delta_1,\ldots,\delta_n}\}_{(U,\delta_1,\ldots,\delta_n,\tau),m})$. Let $\ttt{A}=(M_0,\{B_i\}_i,\{P^{i,U,\delta_1,\ldots,\delta_n,\tau}_m\}_{i,U,\delta_1,\ldots,\delta_n,\tau,m},\Phi)$ be a $1$-to-$s$ $N$-message cloning attack over $\scr{F}$ against $\ttt{Q}^{(3)}$. Then for each $\tau\in S$, $\ttt{A}_\tau=(M_0,\{B_i\}_i,\{P^{i,UV_\tau^\ast,\delta_1,\ldots,\delta_n,\tau}_m\}_{i,U,\delta_1,\ldots,\delta_n,m},\Phi)$ is a $1$-to-$s$ $N$-message cloning attack over $\scr{F}$ against $\ttt{Q}^{(2)}$, so
    \begin{align*}
        \mfk{c}^N_{1\rightarrow s}&(\ttt{Q}^{(3)}|\ttt{A})\\
        &=\int_{\mc{U}(rn)}\int_{D_{rn}^n}\int_S\frac{1}{N}\sum_{m\in M_0}\Tr\squ{(P^{1,U,\delta_1,\ldots,\delta_n,\tau}_m\otimes\cdots\otimes P^{s,U,\delta_1,\ldots,\delta_n,\tau}_m)\Phi(UV_\tau\delta_m V_\tau^\ast U^\ast)}d\upsilon(\tau)d\mu(\delta_1,\ldots,\delta_n)dU\\
        &=\int_S\int_{\mc{U}(rn)}\int_{D_{rn}^n}\frac{1}{N}\sum_{m\in M_0}\Tr\squ{(P^{1,UV_\tau^\ast,\delta_1,\ldots,\delta_n,\tau}_m\otimes\cdots\otimes P^{s,UV_\tau^\ast,\delta_1,\ldots,\delta_n,\tau}_m)\Phi(U\delta_mU^\ast)}d\mu(\delta_1,\ldots,\delta_n)dUd\upsilon(\tau)\\
        &=\int_S\mfk{c}^N_{1\rightarrow s}(\ttt{Q}^{(2)}|\ttt{A}_\tau)d\upsilon(\tau)\leq\mfk{c}^N_{1\rightarrow s}(\ttt{Q}^{(2)}|\scr{F}),
    \end{align*}
    and hence $\mfk{c}^N_{1\rightarrow s}(\ttt{Q}^{(3)}|\scr{F})\leq\mfk{c}^N_{1\rightarrow s}(\ttt{Q}^{(2)}|\scr{F})$.

    To finish, note that if $\delta_m$ is supported on $\mathrm{span}\set{\ket{r(m-1)+1},\ldots,\ket{rm}}$,
    $$\int_S\sigma_m^{UV_\tau,\delta_1,\ldots,\delta_n} d\upsilon(\tau)=U\int V_\tau\delta_m V_\tau^\ast d\upsilon(\tau) U^\ast=U\sigma_m U^\ast.$$
    Now, let $\ttt{A}=(M_0,\{B_i\}_i,\{P^{i,U}_m\}_{i,U,m},\Phi)$ be a $1$-to-$s$ $N$-message cloning attack over $\scr{F}$ against $\ttt{Q}_{r,n}$. Define the $1$-to-$s$ $N$-message cloning attack over $\scr{F}$ against $\ttt{Q}^{(3)}$ $\ttt{A}'=(M_0,\{B_i\}_i,\{P^{i,U}_m\}_{i,U,\delta_1,\ldots,\delta_n,\tau,m},\Phi)$.  Then,
    \begin{align*}
        \mfk{c}^N_{1\rightarrow s}(\ttt{Q}_{r,n}|\ttt{A})&=\int_{\mc{U}(rn)}\frac{1}{N}\sum_{m\in M_0}\Tr\squ{(P^{1,U}_m\otimes\cdots\otimes P^{s,U}_m)\Phi(U\sigma_m U^\ast)}dU\\
        &=\int_S\int_{D_{rn}^n}\int_{\mc{U}(rn)}\frac{1}{N}\sum_{m\in M_0}\Tr\squ{(P^{1,U}_m\otimes\cdots\otimes P^{s,U}_m)\Phi(\sigma^{UV_\tau,\delta_1,\ldots,\delta_n,\tau}_m)}dUd\pi(\delta_1,\ldots,\delta_n)d\upsilon(\tau)\\
        &=\mfk{c}^N_{1\rightarrow s}(\ttt{Q}^{(3)}|\ttt{A}').
    \end{align*}
    Taking suprema, we get that $\mfk{c}^N_{1\rightarrow s}(\ttt{Q}_{r,n}|\scr{F})\leq\mfk{c}^N_{1\rightarrow s}(\ttt{Q}^{(3)}|\scr{F})\leq\mfk{c}^N_{1\rightarrow s}(\ttt{Q}^{(2)}|\scr{F})\leq\mfk{c}^N_{1\rightarrow s}(\ttt{Q}^{(1)}|\scr{F})\leq\mfk{c}^N_{1\rightarrow s}(\ttt{Q}|\scr{F})$.
\end{proof}

\subsection{\texorpdfstring{$t$}{\textit{t}}-copy approximate minimality}

In this section, we show that an approximate form of the minimality of the Haar-measure encryption also holds for attacks using multiple copies of the encrypted state. 

\begin{theorem}\label{thm:multicopy-minimality}
    Let $\ttt{Q}=(M,K,\pi,H,\{\sigma^k_m\}_{k,m})$ be a correct QECM, let $\varepsilon>0$ such that $\int\norm{\sigma^k_m}d\pi(k)\leq\varepsilon$ for each $m$, and let $\scr{F}$ be a class of channels that is closed under $\Phi\mapsto\Phi(V\cdot V^\ast)$ for every isometry $V$. Then for $r=\dim H-|M|+1$ and $n=|M|$, $\mfk{c}^N_{t\rightarrow s}(\ttt{Q}_{r,n}|\scr{F})\leq\mfk{c}^N_{t\rightarrow s}(\ttt{Q}|\scr{F})+7t^2\varepsilon$.
\end{theorem}

\begin{proof}
    The proof proceeds similarly to the proof of \cref{thm:haar-minimality}. First, we enlarge the dimension via an isometry $V:H\rightarrow\C^{rn}$. Let $\ttt{Q}^{(1)}=([n],K,\pi,\C^{rn},\{\tilde{\sigma}^k_m\}_{k,n})$, where $\tilde{\sigma}^k_m=V\sigma^k_m V^\ast$. It is clear that $\mfk{c}_{t\rightarrow s}^N(\ttt{Q}^{(1)}|\scr{F})=\mfk{c}_{t\rightarrow s}^N(\ttt{Q}|\scr{F})$. Next, noting that $\rank\tilde{\sigma}^k_m\leq r$, there exist unitaries $U_k$ and diagonal density matrices $\delta^k_m$ supported on $\mathrm{span}\{\ket{r(m-1)+1},\ldots,\ket{rm}\}$ such that $\tilde{\sigma}^k_m=U_k\delta^k_m U_k^\ast$. 
    
    Then, $\pi$ induces a distribution $\pi'$ on $\mc{U}(rn)\times D_{rn}^n$, where $D_{rn}\subseteq\mc{D}(\C^{rn})$ is the set of diagonal density matrices, as $$\int_{\mc{U}(rn)\times D_{rn}^n}f(U,\delta_1,\ldots,\delta_n)d\pi'(U,\delta_1,\ldots,\delta_n)=\int_kf(U_k,\delta^k_1,\ldots,\delta^k_n)d\pi(k).$$
    Let $\mu$ be the marginal of $\pi'$ on $D_{rn}^n$, let $\pi'_{\delta_1,\ldots,\delta_n}$ be the conditional distribution of $\pi'$ on $\mc{U}(rn)$ given $\delta_1,\ldots,\delta_n$, and define $\sigma^{U,\delta_1,\ldots,\delta_n}_m=U\delta_m U^\ast$. Define the QECM $\ttt{Q}^{(2)}=([n],\mc{U}(rn)\times D_{rn}^n,\mu_{Haar}\times\mu,\C^{rn},\{\sigma^{U,\delta_1,\ldots,\delta_n}_m\}_{(U,\delta_1,\ldots,\delta_n),m})$. This is again a correct QECM. Let $\ttt{A}=(M_0,\{B_i\}_i,\{P^{i,U,\delta_1,\ldots,\delta_n}_m\}_{i,U,\delta_1,\ldots,\delta_n,m},\Phi)$ be a $t$-to-$s$ $N$-message cloning attack over $\scr{F}$ against $\ttt{Q}^{(2)}$. For each $U\in\mc{U}(rn)$, define the $1$-to-$s$ $N$-message cloning attack over $\scr{F}$ against $\ttt{Q}^{(1)}$ as $\ttt{A}_{U}=(M_0,\{B_i\},\{P^{i,UU_k,\delta^k_1,\ldots,\delta^k_n}_{m}\}_{i,k,m},\Phi_U)$, where $\Phi_U(\rho)=\Phi(U^{\otimes t}\rho (U^\ast)^{\otimes t})$. Using Haar invariance, we find that
    \begin{align*}
        &\mfk{c}^N_{t\rightarrow s}(\ttt{Q}^{(2)}|\ttt{A})=\int_{\mc{U}(rn)}\int_{D_{rn}^n}\frac{1}{N}\sum_{m\in M_0}\Tr\squ{(P^{1,U,\delta_1,\ldots,\delta_n}_m\otimes\cdots\otimes P^{s,U,\delta_1,\ldots,\delta_n}_m)\Phi((U\delta_m U^\ast)^{\otimes t})}d\mu(\delta_1,\ldots,\delta_n)dU\\
        &=\int_{\mc{U}(rn)}\int_{D_{rn}^n}\int_{\mc{U}(rn)}\frac{1}{N}\sum_{m\in M_0}\Tr\squ{(P^{1,U,\delta_1,\ldots,\delta_n}_m\otimes\cdots\otimes P^{s,U,\delta_1,\ldots,\delta_n}_m)\Phi((U\delta_m U^\ast)^{\otimes t})}d\pi'_{\delta_1,\ldots,\delta_n}(V)d\mu(\delta_1,\ldots,\delta_n)dU\\
        &=\int_{\mc{U}(rn)}\int_{\mc{U}(rn)\times D_{rn}^n}\frac{1}{N}\sum_{m\in M_0}\Tr\squ{(P^{1,UV,\delta_1,\ldots,\delta_n}_m\otimes\cdots\otimes P^{s,UV,\delta_1,\ldots,\delta_n}_m)\Phi((UV\delta_m V^\ast U^\ast)^{\otimes t})}d\pi'(V,\delta_1,\ldots,\delta_n)dU\\
        &=\int_{\mc{U}(rn)}\int_{K}\frac{1}{N}\sum_{m\in M_0}\Tr\squ{(P^{1,UU_k,\delta^k_1,\ldots,\delta^k_n}_m\otimes\cdots\otimes P^{s,UU_k,\delta_1^k,\ldots,\delta^k_n}_m)\Phi_U((U_k\delta^k_m U_k^\ast )^{\otimes t})}d\pi(k)dU\\
        &=\int_{\mc{U}(rn)}\mfk{c}^N_{t\rightarrow s}(\ttt{Q}^{(1)}|\ttt{A}_U)dU\leq\mfk{c}^N_{t\rightarrow s}(\ttt{Q}^{(1)}|\scr{F}),
    \end{align*}
    so $\mfk{c}^N_{t\rightarrow s}(\ttt{Q}^{(2)}|\scr{F})\leq\mfk{c}^N_{t\rightarrow s}(\ttt{Q}^{(1)}|\scr{F})$.

    Next, let $G\cong\mc{U}(r)^n$ be the subgroup of $\mc{U}(rn)$ that preserves the $\sigma_m$ under conjugation. Define the correct QECM $\ttt{Q}^{(3)}=([n],\mc{U}(rn)\times D_{rn}^n\times G,\mu_{Haar}\times\mu\times\mu_{Haar},\C^{rn},\{\sigma_m^{UV,\delta_1,\ldots,\delta_n}\}_{(U,\delta_1,\ldots,\delta_n,V),m})$. Let $\ttt{A}=(M_0,\{B_i\}_i,\{P^{i,U,\delta_1,\ldots,\delta_n,V}_m\}_{i,U,\delta_1,\ldots,\delta_n,\tau,m},\Phi)$ be a $t$-to-$s$ $N$-message cloning attack over $\scr{F}$ against $\ttt{Q}^{(3)}$. Then for each $V\in G$, $\ttt{A}_V=(M_0,\{B_i\}_i,\{P^{i,UV^\ast,\delta_1,\ldots,\delta_n,V}_m\}_{i,U,\delta_1,\ldots,\delta_n,m},\Phi)$ is a $t$-to-$s$ $N$-message cloning attack over $\scr{F}$ against $\ttt{Q}^{(2)}$, so
    \begin{align*}
        \mfk{c}^N_{t\rightarrow s}&(\ttt{Q}^{(3)}|\ttt{A})\\
        &=\int_{\mc{U}(rn)}\int_{D_{rn}^n}\int_G\frac{1}{N}\sum_{m\in M_0}\Tr\squ{(P^{1,U,\delta_1,\ldots,\delta_n,V}_m\otimes\cdots\otimes P^{s,U,\delta_1,\ldots,\delta_n,V}_m)\Phi((UV\delta_m V^\ast U^\ast)^{\otimes t})}dVd\mu(\delta_1,\ldots,\delta_n)dU\\
        &=\int_G\int_{\mc{U}(rn)}\int_{D_{rn}^n}\frac{1}{N}\sum_{m\in M_0}\Tr\squ{(P^{1,UV^\ast,\delta_1,\ldots,\delta_n,V}_m\otimes\cdots\otimes P^{s,UV^\ast,\delta_1,\ldots,\delta_n,V}_m)\Phi((U\delta_mU^\ast)^{\otimes t})}d\mu(\delta_1,\ldots,\delta_n)dUdV\\
        &=\int_G\mfk{c}^N_{t\rightarrow s}(\ttt{Q}^{(2)}|\ttt{A}_V)dV\leq\mfk{c}^N_{t\rightarrow s}(\ttt{Q}^{(2)}|\scr{F}),
    \end{align*}
    and hence $\mfk{c}^N_{t\rightarrow s}(\ttt{Q}^{(3)}|\scr{F})\leq\mfk{c}^N_{t\rightarrow s}(\ttt{Q}^{(2)}|\scr{F})$.

    To finish, we make use of \cref{lem:haar_random_encryption_scheme_t_copy_untelegraphable_indistinguishable_lemma1}. Fix $m$, and let $T:\C^{r}\rightarrow\C^{rn}$ be the isometry $T\ket{i}=\ket{i+r(m-1)}$. Then, for $(\delta_1,\ldots,\delta_n)$ in the support of $\mu$, $\delta_m$ is supported on the image of $T$. Hence we have that
    \begin{align*}
        (T^\ast)^{\otimes t}\int_{G}(V\delta_mV^\ast)^{\otimes t} dV T^{\otimes t}=\int_{\mc{U}(r)}U^{\otimes t}(T^\ast\delta_m T)^{\otimes t}U^{\otimes t}dU.        
    \end{align*}
    Let $\rho=T^\ast\delta_m T$. By \cref{lem:haar_random_encryption_scheme_t_copy_untelegraphable_indistinguishable_lemma1}, we have that
    \begin{align*}
        \norm*{\int_{\mc{U}(r)}(U\rho U^\ast)^{\otimes t}dU-\frac{I}{r^t}}_{\Tr}\leq 7t^2\norm{\rho}.
    \end{align*}
    Extending outside the image of $T$ (where all the states are $0$), we find that
    \begin{align*}
        \norm*{\int_{G}(V\delta_mV^\ast)^{\otimes t}dV-\sigma_m^{\otimes t}}_{\Tr}\leq 7t^2\norm{\delta_m}.
    \end{align*}
    Now, let $\ttt{A}=(M_0,\{B_i\}_i,\{P^{i,U}_m\}_{i,U,m},\Phi)$ be a $t$-to-$s$ $N$-message cloning attack over $\scr{F}$ against $\ttt{Q}_{r,n}$. Define the $t$-to-$s$ $N$-message cloning attack over $\scr{F}$ against $\ttt{Q}^{(3)}$ $\ttt{A}'=(M_0,\{B_i\}_i,\{P^{i,U}_m\}_{i,U,\delta_1,\ldots,\delta_n,V,m},\Phi)$.  Then,
    \begin{align*}
        \mfk{c}^N_{t\rightarrow s}(\ttt{Q}_{r,n}|\ttt{A})&=\int_{\mc{U}(rn)}\frac{1}{N}\sum_{m\in M_0}\Tr\squ{(P^{1,U}_m\otimes\cdots\otimes P^{s,U}_m)\Phi((U\sigma_m U^\ast)^{\otimes t})}dU\\
        &=\int_G\int_{D_{rn}^n}\int_{\mc{U}(rn)}\frac{1}{N}\sum_{m\in M_0}\Tr\squ{(P^{1,U}_m\otimes\cdots\otimes P^{s,U}_m)\Phi(\sigma^{UV,\delta_1,\ldots,\delta_n,V}_m)}dUd\mu(\delta_1,\ldots,\delta_n)dV\\
        &\qquad+\int_{D_{rn}^n}\frac{1}{N}\sum_{m\in M_0}7t^2\norm{\delta_m}d\mu(\delta_1,\ldots,\delta_n)\\
        &=\mfk{c}^N_{t\rightarrow s}(\ttt{Q}^{(3)}|\ttt{A}')+7t^2\varepsilon.
    \end{align*}
    Taking suprema, we get that $\mfk{c}^N_{t\rightarrow s}(\ttt{Q}_{r,n}|\scr{F})\leq\mfk{c}^N_{t\rightarrow s}(\ttt{Q}^{(3)}|\scr{F})+7t^2\varepsilon$.
\end{proof}

Using \cref{thm:multicopy-minimality} along with Corollary 3.3 from \cite{MST21arxiv}, we can get an approximate version of the minimality for $t$-to-$t+1$-copy security uncloneable encryption.

\begin{corollary}
    Let $\ttt{Q}=(M,K,\pi,H,\{\sigma^k_m\}_{k,m})$ be a correct QECM and write $r=\dim H-|M|+1$ and $n=|M|$. Then, if $\mfk{c}^N_{t\rightarrow t+1}(\ttt{Q}_{r,n})=\frac{1}{N}+\delta$, then $\mfk{c}^N_{t\rightarrow t+1}(\ttt{Q})\geq\frac{1}{N}+\frac{\delta}{56 Nt^2+1}$.
\end{corollary}

\begin{proof}
    Let $\varepsilon=\max_{m\in M}\int_K\norm{\sigma^k_m}d\pi(k)$. Then, by Corollary 3.3 from \cite{MST21arxiv}, $\mfk{c}_{1\rightarrow 2}^2(\ttt{Q})\geq\frac{1}{2}+\frac{\varepsilon}{16}$. Now, using~\cite[Theorem 12]{BL20} as in \cref{cor:lower-from-dist-teleg}, $\mfk{c}^N_{1\rightarrow 2}(\ttt{Q})\geq\frac{2}{N}\mfk{c}_{1\rightarrow 2}^2(\ttt{Q})\geq\frac{1}{N}+\frac{\varepsilon}{8N}$. This also provides a lower bound for $\mfk{c}^N_{t\rightarrow t+1}(\ttt{Q})$ by using the strategy where the first $t-1$ players get a copy of the ciphertext state and win perfectly, and the final two players play with the $1$-to-$2$ cloning attack. Hence, if $\varepsilon\geq\frac{\delta}{7t^2+\frac{1}{8N}}$, then $\mfk{c}^N_{t\rightarrow t+1}(\ttt{Q})\geq\frac{1}{N}+\frac{\delta}{56 Nt^2+1}$.
    
    On the other hand, we know by \cref{thm:multicopy-minimality} that $\mfk{c}^N_{t\rightarrow t+1}(\ttt{Q})\geq\frac{1}{N}+\delta-7t^2\varepsilon$. Hence, we get the other case: if $\varepsilon\leq\frac{\delta}{7t^2+\frac{1}{8N}}$, then $\mfk{c}^N_{t\rightarrow t+1}(\ttt{Q})\geq\frac{1}{N}+\frac{\delta}{56 Nt^2+1}$.
\end{proof}

\subsection{Implications of minimality}

We can put together the results of this section with the lower bounds of \cref{sec:lowerbounds} to get lower bounds that hold for any QECM.

\begin{corollary} \label{cor:UTE_and_UE_lower_bounds}
    Let $\ttt{Q}=(M,K,\mu,H,\{\sigma^k_m\}_{k,m})$ be a correct QECM. Then, writing $d=\dim H$,
    \begin{align*}
        &c^{N}_{1\rightarrow 1}(\ttt{Q}|\scr{M})\geq\frac{1}{N}+\frac{1}{N\sqrt{\pi(d-|M|+1)}}+O\parens*{\frac{1}{N(d-|M|+1)^{3/2}}}\\
        &c^N_{1\rightarrow 2}(\ttt{Q})\geq\frac{1}{N}+\frac{1}{N\sqrt{\pi(d-|M|+1)}}+O\parens*{\frac{1}{N(d-|M|+1)^{3/2}}}\\
        &c^N_{t\rightarrow t+1}(\ttt{Q})\geq \frac{1}{N}+\frac{1}{57N^2\sqrt{\pi t^3(d-|M|+1)}}+O\parens*{\frac{1}{N^2(t(d-|M|+1))^{3/2}}}.
    \end{align*}
\end{corollary}

In particular, previous work has shown that the cloning-distinguishing value is lower-bounded as $\mfk{c}^2_{1\rightarrow 2}(\ttt{Q})=\frac{1}{2}+\Omega\parens*{\frac{1}{d}}$~\cite{MST21arxiv}, which we strengthen to $\mfk{c}^2_{1\rightarrow 2}(\ttt{Q})=\frac{1}{2}+\Omega\parens*{\frac{1}{\sqrt{d}}}$.

\section*{Acknowledgements} The authors would like to thank Archishna Bhattacharyya and Arthur Mehta for helpful discussions.  We acknowledge the support of the Natural Sciences and Engineering Research Council of Canada (NSERC)(ALLRP-578455-2022), 
the Air Force Office of Scientific Research under award number FA9550-20-1-0375 and of the Canada Research Chairs Program.
\fi

\bibliographystyle{bibtex/bst/alphaarxiv.bst}
\bibliography{bibtex/bib/quasar-full.bib,
              bibtex/bib/quasar.bib,
              bibtex/bib/quasar-more-merged.bib,
                bibtex/bib/quasar-more.bib,
              bibtex/bib/quasar-more-telegraph-new.bib}
\end{document}